\documentclass[journal,12pt,onecolumn,draftclsnofoot]{IEEEtran}
\usepackage{ifpdf}

\ifCLASSINFOpdf
 \usepackage[pdftex]{graphicx}
\else
\fi
\usepackage{amsmath}
\usepackage{amsfonts}
\usepackage{amssymb}
\usepackage{amsthm}
\usepackage{algorithmic}
\usepackage{url}

\usepackage{array}
\usepackage{float}
\usepackage{url}
\usepackage[usenames,dvipsnames,svgnames,table]{xcolor}

\newtheorem{theorem}{Theorem}
\newtheorem{definition}{Definition}
\newtheorem{example}{Example}

\newtheorem{corollary}{Corollary}
\newtheorem{lemma}{Lemma}
\newtheorem{remark}{Remark}

\begin{document}
%
\title{Multilevel constructions: coding, packing and geometric uniformity}
%
%
%

\author{Maiara~F.~Bollauf, 
        Ram~Zamir 
        and Sueli~I.~R.~Costa
\thanks{M. F. Bollauf was with the Institute of Mathematics, Statistics and Computer Science, University of Campinas as a Ph.D student and now she is with Dep. of Electrical and Computer Engineering, Texas A$\&$M University at Qatar (e-mail: maiara.bollauf@qatar.tamu.edu).}
\thanks{R. Zamir is with Dep. of Electrical Engineering-Systems in Tel Aviv University (e-mail: zamir@eng.tau.ac.il).}
\thanks{S. I. R. Costa is with the Institute of Mathematics, Statistics and Computer Science, University of Campinas (e-mail: sueli@ime.unicamp.br).}
\thanks{Partial and preliminary results of this paper were presented at IEEE International Symposium on
Information Theory, Barcelona, Spain, 2016 \cite{bollaufzamir} and at International Zurich Seminar, Zurich, Switzerland, 2018 \cite{bollaufzamircosta}.}}

\maketitle

\begin{abstract}

Lattice and special nonlattice multilevel constellations constructed from binary codes, such as Constructions A, C, and D, have relevant applications in Mathematics (sphere packing) and in Communication (multi-stage decoding and efficient vector quantization). In this work, we explore some properties of Construction C, in particular its geometric uniformity. We then propose a new multilevel construction, inspired by bit interleaved coded modulation (BICM), that we call Construction $C^\star$. We investigate the geometric uniformity, laticeness, and minimum distance properties of Construction $C^\star$ and discuss its superior packing efficiency when compared to Construction C.

\end{abstract}


%
\IEEEpeerreviewmaketitle

\section{Introduction}

	Lattices are discrete sets in $\mathbb{R}^n$ which are described as all integer linear combinations of a set of independent vectors. Some efficient lattice constructions are based on linear codes, a topic that have been stimulated by the comprehensive approach in Conway and Sloane \cite{conwaysloane}. Construction $C$ \cite{lee71} is a multilevel code construction \cite{imaihirakawa, ungerboeck}, which if based on linear codes satisfying certain nesting relations, forms a lattice.

	When we consider a single level with a linear code, Construction C is the well known lattice Construction $A$ \cite{conwaysloane}. However, when the number of levels $L > 1,$ the resulting construction is not always a lattice, even if the component codes are linear. The work of Kositwattanarerk and Oggier \cite{kositoggier} explored the relation between Construction C and the lattice Construction D. They showed that if we consider a family of nested linear binary codes $\mathcal{C}_{1} \subseteq \dots \subseteq \mathcal{C}_{L} \subseteq \mathbb{F}_{2}^{n},$ where this chain is closed under Schur product, then both constructions coincide and we obtain a lattice (an extension of this result for  codes over a $q-$ary alphabet can be found in \cite{streycosta}). 
Regarding general properties, for linear component codes case, all points of the constellation defined by Construction C have the same minimum distance \cite{conwaysloane}, but not necessarily the same kissing number (see Example \ref{cex2}). A notable example of a lattice Construction C is the Barnes-Wall lattice, generated by the family of Reed-Muller codes \cite{forney1}, \cite{forney}.

	Through multi-stage decoding, Construction C can achieve the high SNR uniform-input capacity of an additive white Gaussian noise (AWGN) channel asymptotically as the dimension $n$ goes to infinity \cite{forneytrottchung}. Another interesting example of nonlattice construction is presented by Agrell and Eriksson in \cite{agrelleriksson}, who proved that the ``$D_{n}+$" tessellation \cite{conwaysloane}, which is a $2-$level Construction C, exhibits a lower normalized second moment (and then a better quantization efficiency) than any known lattice tessellation in dimensions $7$ and $9.$ 
	
	Our first objective in this paper is to study the properties of a general Construction C with underlying binary codes, and to find out how close to a lattice this construction can be, in case it does not satisfy the conditions required in \cite{kositoggier}. We show that a two-level ($L=2$) Construction C is geometrically uniform (a result that can be deduced from \cite{forneyguc}), however for three levels or more ($L\geq 3$) the distance spectrum may vary between the points of the constellation. 
	
	A recent coded modulation scheme, referred to as \textit{bit interleaved coded modulation} (BICM) \cite{alvarado}, motivates the second contribution of this paper: the proposal and study of a new multilevel construction, Construction $C^\star,$ which arises as a generalization of Construction C.  
	 
%
		
	The BICM, first introduced by Zehavi \cite{zehavi}, requires mainly to have: a length$-nL$ binary code $\mathcal{C},$ an interleaver (permutation) $\pi,$ and a one-to-one binary labeling map $\tilde{\mu}: \{0,1\}^{L} \rightarrow \mathcal{X},$ where $\mathcal{X}$ is a signal set $\mathcal{X}=\{0,1, \dots, 2^{L}-1\}$ (alphabet). The code and interleaved bit sequence $c \in \mathcal{C}$ is partitioned into $L$ subsequences $c_i,$ $i=1, \dots, L,$ of length $n:$
\begin{equation*}
c=(c_{1}, \dots, c_{L}), \ \ \mbox{with} \ \ c_{i}=(c_{i1},c_{i2}, \dots, c_{in}). 
\end{equation*}
	
	The bits $c_{j}$ are mapped at a time index $j$ to a symbol $x_{j}$ chosen from the $2^{L}-$ary signal constellation $\mathcal{X},$ according to the binary labeling map $\tilde{\mu}.$ Hence, for a length$-nL$ binary code $\mathcal{C}$ to encode all bits, we have the scheme below:
	
{\small \begin{center}

$\boxed{\text{codeword} \ c}$ $\rightarrow$ $\boxed{\text{interleaver} \ \pi }$ $\rightarrow$ $\boxed{\text{partitioning into} \ $L$ \ \text{subsequences of length} \ $n$}$ $\rightarrow$ $\boxed{\text{mapping} \ \tilde{\mu}}$ 
\vspace{0.3cm}
$\rightarrow$ $\boxed{x_{j}=\tilde{\mu}(c_{1j}, \dots, c_{Lj}), \ j={1, \dots, n}}$

\end{center}}
	
	By defining the natural labeling $\mu: \mathcal{C} \rightarrow \mathcal{X}^{n}$ as $\mu(c_{1}, c_{2},, \dots, c_{L})=c_{1}+2c_{2}+\dots+2^{L-1}c_{L}$ and assuming $\pi(\mathcal{C})=\mathcal{C},$ it is possible to define an extended BICM constellation in a way very similar to the multilevel Construction C, which we denote by Construction $C^{\star}.$
	
	The constellation produced via Construction $C^\star$ is always a subset of the associated constellation produced via Construction C for the same projection codes and it does not usually results in a lattice. We explore here some facets of this original construction by mainly presenting a necessary and sufficient condition for it to be a lattice, and also describing the Leech lattice $\Lambda_{24}$ via Construction $C^\star.$ Besides that, we study some properties of Construction $C^\star,$ such as geometric uniformity and minimum distance, in order to analyze and compare packing efficiencies of Construction $C^\star$ and Construction C.
	
	This paper is organized as follows. Section II is devoted to preliminary concepts and results. In Section III we point out known properties of Construction C and present a detailed discussion about its geometric uniformity. In Section IV we exhibit general geometrically uniform constellations and as a consequence, an alternative proof to the geometric uniformity of an $L=2$ Construction C. In Section V we introduce Construction $C^{\star}$ and illustrate it with examples, including a special characterization of the Leech lattice. In Section VI we investigate properties of Construction $C^\star$ such as geometric uniformity and latticeness. Section VII is devoted to the study of the minimum distance of a constellation defined by Construction $C^\star$ and packing density relations between Constructions C and $C^\star.$ Section VIII brings an asymptotic analysis of a random Construction $C^\star$ comparing its packing efficiency with the best known packing efficiency of Construction C and finally Section IX concludes the paper.


\section{Lattice basics and constructions from linear codes}

	This section is devoted to present some basic concepts, results and notations \cite{conwaysloane, zamir2014} to be used in the next sections.

	We will denote by $+$ the real addition and by $\oplus$ the sum in $\mathbb{F}_{2},$ i.e., $x \oplus y=(x+y) \bmod 2.$
	
	 A \textit{linear binary code} $\mathcal{C}$ of length $n$ and rank $k$ is a linear subspace of dimension $k$ over the vector space $\mathbb{F}_{2}^{n}.$ It can also be written as the image of an injective linear transformation $\phi :  \mathbb{F}_{2}^{k} \rightarrow \mathbb{F}_{2}^{n},$ $(a_{1}, \dots, a_{k}) \mapsto  G_{\mathcal{C}} \cdot (a_{1}, \dots, a_{k})^{T},$ where $G_{\mathcal{C}} \in \mathbb{F}_{2}^{n \times k}$ is a \textit{generator matrix} of $\mathcal{C}.$ The columns of $G_{\mathcal{C}}$ compose a basis for $\mathcal{C}.$

	The \textit{Hamming distance} between two elements $x=(x_{1}, \dots, x_{n}),y=(y_{1}, \dots, y_{n}) \in \mathbb{F}_{2}^{n}$ is the number of different symbol in the two codewords,
	\begin{eqnarray}
	d_{H}(x,y)&=& |\{i: x_{i} \neq y_{i}, 1 \leq i \leq n\}|.
	\end{eqnarray}
The \textit{Hamming weight} $\omega(c)$ is the number of nonzero elements in a codeword $c \in \mathcal{C}.$
	
	The \textit{minimum distance} of a binary code $\mathcal{C}$ is the minimum Hamming distance between all distinct codewords, i.e.,
\begin{equation}
d_{H}(\mathcal{C})=d_{\min}(\mathcal{C})=\min \{d_{H}(x,y): x,y \in \mathcal{C}, x \neq y\}.
\end{equation}
For a linear code, the minimum distance is also the minimum Hamming weight of a nonzero codeword.

	A linear binary code of length $n$ and rank $k,$ with $2^k$ codewords and minimum distance $d=d_{\min}(\mathcal{C})$ is said to be an $(n,k,d)-$code. The \textit{rate} of such a code is 
\begin{equation}
	R=\frac{1}{n} \log_{2} 2^k= \frac{k}{n}\ \text{bits/symbol}.
\end{equation}
	


	A \textit{lattice} $\Lambda \subset \mathbb{R}^{N}$ is the set of all integer linear combinations of independent vectors $v_{1}, v_{2}, \dots, v_{n} \in \mathbb{R}^{N},$ $n \leq N$ and $\{v_1, v_2, \dots, v_n\}$ is called a basis of $\Lambda.$ A matrix $G_{\Lambda} \in \mathbb{R}^{N \times n},$ whose columns compose a basis of $\Lambda,$ is called \textit{generator matrix} of $\Lambda.$ In other words, a lattice is a discrete additive subgroup of $\mathbb{R}^{n}.$ We consider here only full rank $(n=N)$ lattices.

%

	For a lattice $\Lambda \subset \mathbb{R}^{n},$ the minimum distance is the smallest Euclidean distance between any two lattice points
\begin{eqnarray}
d_{E}(\Lambda)=d_{\min}(\Lambda)=\inf \{||x-y||: x,y \in \Lambda, x \neq y\}.
\end{eqnarray} 

	
	The \textit{Voronoi region} $\mathcal{V}(\lambda)$ of a lattice point $\lambda \in \Lambda$ is the open subset of $\mathbb{R}^{n}$ containing all points nearer to $\lambda$  than to any other lattice point. The closure of a Voronoi region tesselates $\mathbb{R}^{n}$ by translations given by lattice points.
	
	The \textit{packing radius} $r_{\text{pack}}(\Lambda)$ of a lattice $\Lambda$ is half of the minimum distance between lattice points and the \textit{packing density} $\Delta(\Lambda)$ is the fraction of  space that is covered by balls $\mathcal{B}(\lambda, r_{\text{pack}}(\Lambda))$ of radius $r_{\text{pack}}(\Lambda)$ in $\mathbb{R}^{n}$ centered at lattice points $\lambda \in \Lambda,$ i.e., 
\begin{equation}
\Delta(\Lambda)  =  \dfrac{vol \ B(0,r_{\text{pack}}(\Lambda))}{vol (\Lambda)},
\end{equation}
where $vol(\Lambda)=|\det(G_\Lambda)|=vol(\mathcal{V}(\lambda)).$  The \textit{effective radius} $r_{\text{eff}}(\Lambda)$ is the radius of a ball with the same volume, i.e., $vol(\Lambda)=V_n r_{\text{eff}}^n(\Lambda),$ where $V_n$ denotes the volume of the unit ball. The \textit{packing efficiency} is defined as 
\begin{equation}\label{eqpackeff}
\rho_{\text{pack}}(\Lambda) = \dfrac{r_{\text{pack}}(\Lambda)}{r_{\text{eff}}(\Lambda)}=(\Delta(\Lambda))^{1/n}.
\end{equation}

%
%

	A constellation\footnote{A constellation is a discrete set of points in $\mathbb{R}^n$.} $\Gamma \subset \mathbb{R}^{n}$ is said to be \textit{geometrically uniform} \cite{forneyguc} if for any two elements $c,c' \in \Gamma$ there exists a distance-preserving transformation $T$ such that $c'=T(c)$ and $T(\Gamma)=\Gamma.$
	Every lattice $\Lambda$ is geometrically uniform, due to the fact that any translation $\Lambda+x$ by a lattice point $x \in \Lambda$ is just $\Lambda$ and this implies that every point of the lattice has the same number of neighbors at each distance and all Voronoi regions are congruent, $\mathcal{V}(\lambda)=\mathcal{V}(0)+\lambda.$ 

	From linear codes, it is possible to derive lattice constellations using the well known Constructions $A$ and $D$ \cite{conwaysloane}.
	
\begin{definition}(Construction A) \label{constrA} Let $\mathcal{C}$ be an $(n,k,d)-$linear binary code. We define Construction A from $\mathcal{C}$\footnote{Construction A can be defined also using a nonlinear code, but then the resulting constellation is not a lattice.} as
\begin{equation}
\Lambda_{A}=\mathcal{C} + 2\mathbb{Z}^{n}.
\end{equation}
\end{definition}

\begin{definition} (Construction D)  Let $\mathcal{C}_{1} \subseteq \dots \subseteq \mathcal{C}_{L} \subseteq \mathbb{F}_{2}^{n}$ be a family of nested linear binary codes. 
	Let $k_{i}=\dim(\mathcal{C}_{i})$ and let $\{b_{1}, b_{2}, \dots, b_{n}\}$ be a basis of $\mathbb{F}_{2}^{n}$ such that $\{b_{1}, \dots, b_{k_{i}}\}$ span $\mathcal{C}_{i}.$ The lattice $\Lambda_D$ consists of all vectors of the form
	\begin{equation}
	\Lambda_{D}= \displaystyle\sum_{i=1}^{L} 2^{i-1} \displaystyle\sum_{j=1}^{k_{i}} \alpha_{j}^{(i)} \psi(b_{j})+2^{L}z
	\end{equation}
	where $\alpha_{j}^{(i)} \in \{0,1\},$ $z \in \mathbb{Z}^{n}$ and $\psi$ is the natural embedding of $\mathbb{F}_{2}^n$ in $\mathbb{Z}^n.$ 
\end{definition}

	Another well studied construction, that in general does not produce lattice constellation, even when the underlying codes are linear is Construction C, defined below following the Code-Formula terminology from \cite{forney1} (more details and applications also in  \cite{agrelleriksson}). Observe that this particular case does not enforce the minimum distance conditions imposed in \cite[pp. 150]{conwaysloane}. 
	
\begin{definition}(\textit{Construction C})  Consider $L$ binary codes $\mathcal{C}_{1}, \dots, \mathcal{C}_{L} \subseteq \mathbb{F}_{2}^{n},$ not necessarily nested or linear. The infinite constellation $\Gamma_{C}$ in $\mathbb{R}^{n},$ called Construction C, is defined as:
	\begin{equation} \label{eqC}
	\Gamma_{C}=\mathcal{C}_{1}+2\mathcal{C}_{2}+ \dots + 2^{L-1}\mathcal{C}_{L}+2^{L}\mathbb{Z}^{n},
	\end{equation} 
	or equivalently
	\begin{eqnarray}
	\Gamma_{C} &:=& \{c_1 + 2c_2 + \dots + 2^{L-1} c_L + 2^L z: c_i \in \mathcal{C}_i, i=1,\dots,L, \ z \in \mathbb{Z}^n\}.
	\end{eqnarray}
\end{definition}

	Note that if $L=1,$ i.e., if we consider a single level with a linear code, then this construction reduces to lattice Construction A (Definition \ref{constrA}). 


\begin{example}(A nonlattice Construction C) Consider $\mathcal{C}_{1}=\{(0,0),(1,1)\}$ and $\mathcal{C}_{2}=\{(0,0)\}.$ The $2-$level Construction $C$ from these codes is given by $\Gamma_{C}=\mathcal{C}_{1}+2\mathcal{C}_{2}+4\mathbb{Z}^{2}.$ Geometrically, we can see this constellation in Figure \ref{ex1c} and clearly $\Gamma_{C}$ is not a lattice.
	
	\begin{figure}[H]
		\begin{center}
			\includegraphics[height=5.1cm]{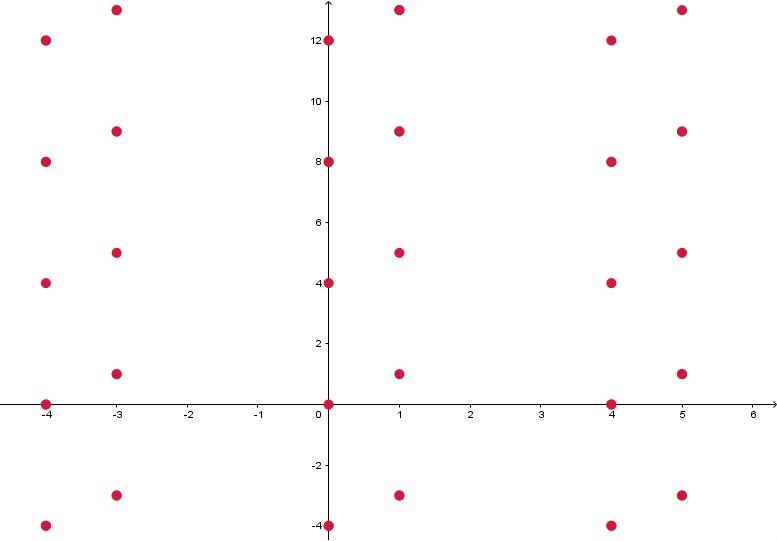}  
		   	\caption{{$2$- level Construction C from $\mathcal{C}_{1}=\{(0,0),(1,1)\}$ and $\mathcal{C}_{2}=\{(0,0)\}$}.} \label{ex1c}
		\end{center}
	\end{figure}
	
\end{example}

\begin{definition} (\textit{Schur product}) For $x=(x_{1}, \dots, x_{n})$ and $y=(y_{1}, \dots, y_{n})$ both in $\mathbb{F}_{2}^{n},$ we define $x \ast y = (x_{1}y_{1}, \dots, x_{n}y_{n}).$
\end{definition}

	Consider $\psi: \mathbb{F}_2^{n} \rightarrow \mathbb{R}^{n}$ as the natural embedding. Then, for $x, y \in \mathbb{F}_{2}^{n},$ it is valid that
\begin{equation}\label{eqss}
\psi(x)+\psi(y)=\psi(x+y)+2\psi(x \ast y).
\end{equation} 
In order to simplify, we abuse the notation, writing Equation \eqref{eqss} as	
\begin{equation} \label{sum}
x+y=x \oplus y + 2(x \ast y).
\end{equation}

	Kositwattanarerk and Oggier \cite{kositoggier} provided a condition that, if satisfied, guarantees that Construction C will provide a lattice which coincides with Construction D.

\begin{theorem} \cite{kositoggier} \label{thmko} (\textit{Latticeness of Construction C}) Given a family of nested binary linear codes $\mathcal{C}_{1} \subseteq \dots \subseteq \mathcal{C}_{L} \subseteq \mathbb{F}_{2}^{n},$ then the following statements are equivalent:
	\begin{itemize}[\IEEEsetlabelwidth{Z}]
		\item[1.] $\Gamma_{C}$ is a lattice.
		\item[2.] $\mathcal{C}_{1} \subseteq \dots \subseteq \mathcal{C}_{L} \subseteq \mathbb{F}_{2}^{n}$ is closed under Schur product.
		\item[3.] $\Gamma_{C}=\Lambda_{D}.$ 
	\end{itemize}
\end{theorem} 


\begin{example}($D_{n}+$ as Construction C) The $D_{n}+$ tessellation \cite{conwaysloane} can be written as a $2-$level Construction C if we consider $\mathcal{C}_{1}$ as the $(n,1,n)-$repetition code and $\mathcal{C}_{2}$ as the $(n,n-1,2)-$even parity check code. For $n$ even, this construction represents a lattice, because we would have nested linear codes that are closed under Schur product. Otherwise, when $n$ is odd, we obtain a nonlattice constellation Construction C. In particular, for dimensions $n=7$ and $9$ it was proved that $D_{n}+$ has a lower normalized second moment than any known lattice tessellation \cite{agrelleriksson}. 
\end{example}


\section{Properties of Construction C}

	We summarize next some known properties of Construction C already explored in the literature (\ref{secMD}, \ref{seckn}, \ref{subsgu}) and we present a counterexample to show that for three levels and up, Construction C is not geometrically uniform in general (\ref{seceds3}).
	
\subsection{Minimum distance} \label{secMD}

 	If the underlying codes of Construction C are linear, then the squared minimum distance can be expressed \cite[pp. 150]{conwaysloane} as
\begin{eqnarray} \label{dminconstrc}
d_{\min}^{2}(\Gamma_{C})=\min \{d_{H}(\mathcal{C}_{1}), 2^{2}d_{H}(\mathcal{C}_{2}), \dots, 2^{2(L-1)}d_{H}(\mathcal{C}_{L}), 2^{2L}\}.
\end{eqnarray}

	Indeed, observe that sets defined as $\Gamma_{\mathcal{C}_i}=\{0 + 2 \cdot 0 + \dots + 2^{i-1}c_i + \dots + 2^{L-1} \cdot 0 + 2^{L}  \cdot 0,~ c_i \in \mathcal{C}_i\},$ where $0 \in \mathbb{R}^{n},$ are subsets of $\Gamma_{C},$ i.e., $\Gamma_{\mathcal{C}_i} \subseteq \Gamma_{C}$ for all $i=1, \dots, L,$ then it follows that $d_{E_{min}}^{2}(\Gamma_{C}) \leq \min \{d_{H}(\mathcal{C}_{1}), 2^{2}d_{H}(\mathcal{C}_{2}), \dots, 2^{2(L-1)}d_{H}(\mathcal{C}_{L}), 2^{2L}\}.$
	
	On the other hand, according to the discussion in \cite[pp. 150]{conwaysloane}, if the lowest component in which two elements in $\Gamma_{C}$ differ is the $i-$th component, then their squared distance is always greater than or equal to $2^{(i-1)}d_H(\mathcal{C}_i),$ for $i=1, \dots, L.$ Hence, it follows that $d_{E_{min}}^{2}(\Gamma_{C}) \geq \min \{d_{H}(\mathcal{C}_{1}), 2^{2}d_{H}(\mathcal{C}_{2}), \dots, 2^{2(L-1)}d_{H}(\mathcal{C}_{L}),$ $2^{2L}\}$ and we conclude that Equation \eqref{dminconstrc} holds.  

	From the formula of the squared minimum distance, we can also conclude that all points in Construction C attains the same minimum distance. 

\subsection{Kissing number} \label{seckn}

	The kissing number (number of nearest neighbors) of an element of Construction C may vary between the elements even when the underlying codes are linear, as it can be seen in the $3-$level construction in Example \ref{cex2}, Section \ref{seceds3}. It follows, in particular, that the points of Construction C may not have the same distance spectrum, hence it is not geometrically uniform in general.

\subsection{Geometric uniformity} \label{subsgu}

 	The geometric uniformity of a two-level ($L=2$) Construction C, i.e. $\Gamma_C=\mathcal{C}_1+2\mathcal{C}_2+4\mathbb{Z}^n,$ can be deduced from the work of G. D. Forney \cite{forneyguc} if we consider a $2-$level Construction C as a group code with isometric labeling over $\mathbb{Z}/4\mathbb{Z}$ (i.e., a $2-$level binary coset code over $\mathbb{Z}/ 2\mathbb{Z} / 4\mathbb{Z}).$ He proved that this type of construction produces a geometrically uniform generalized coset code.  In Section \ref{sec4}, we provide an alternative proof, based on explicit isometric transformation. 

\subsection{Equi-distance spectrum and geometric uniformity for \ $L\geq 3$} \label{seceds3}

 Geometric uniformity implies, in particular, that all points have the same set of Euclidean distances to their neighbors. 

\begin{definition} (\textit{Distance spectrum}) For a discrete constellation $\Gamma \subset \mathbb{R}^{n},$ the distance spectrum is
	\begin{center}
		$N(c,d)=$ number of points in the constellation at an Euclidean distance $d$ from an element $c$ in the constellation.
\end{center}
\end{definition}

\begin{definition} (\textit{Equi-distance spectrum}) A discrete constellation $\Gamma \subset \mathbb{R}^n$ is said to have equi-distance spectrum (EDS) if $N(c,d)$ is the same for all $c \in \Gamma.$
\end{definition}

	We also introduce here the terminology  of \textit{equi-minimum distance}, where each point in the constellation has at least one neighbor at minimum distance, i.e., for all $c \in \Gamma,$ there exist  $c' \in \Gamma,$ such that $d(c,c')=d_{\min}(\Gamma).$

	Since geometric uniformity implies equi-distance spectrum, in particular for a $2-$level Construction C, $N(c,d)=N(0,d)$ for all $c \in \Gamma_{C}.$ However for $L \geq 3$ the equi-distance spectrum and hence the geometric uniformity property does not hold in general, as we will see in the upcoming example. 

\begin{example}(Non-EDS two dimensional Construction C) \label{cex2} Consider an $n=2$ and $L=3$ Construction C with the following three linear component codes: 
	\begin{eqnarray}
	\mathcal{C}_{1}=\mathcal{C}_{2}=\{(0,0),(1,1)\}, ~ \mathcal{C}_{3}=\{(0,0)\}.
	\end{eqnarray}
We can write $\Gamma_{C}= \mathcal{C}_{1}+2\mathcal{C}_{2}+ 4\mathcal{C}_{3}+8\mathbb{Z}^{3}$ (Figure \ref{p2}) in this case as $\Gamma_{C}=\{(8k_{1}+j, 8k_{2}+j): k_{1}, k_{2} \in \mathbb{Z}, j =0,1,2,3 \}.$ One can notice that $N((3,3),\sqrt{2})=1 \neq 2 = N((1,1),\sqrt{2}),$ so it is not geometrically uniform.
\begin{figure}[h!]
\begin{center}
		\includegraphics[height=5.1cm]{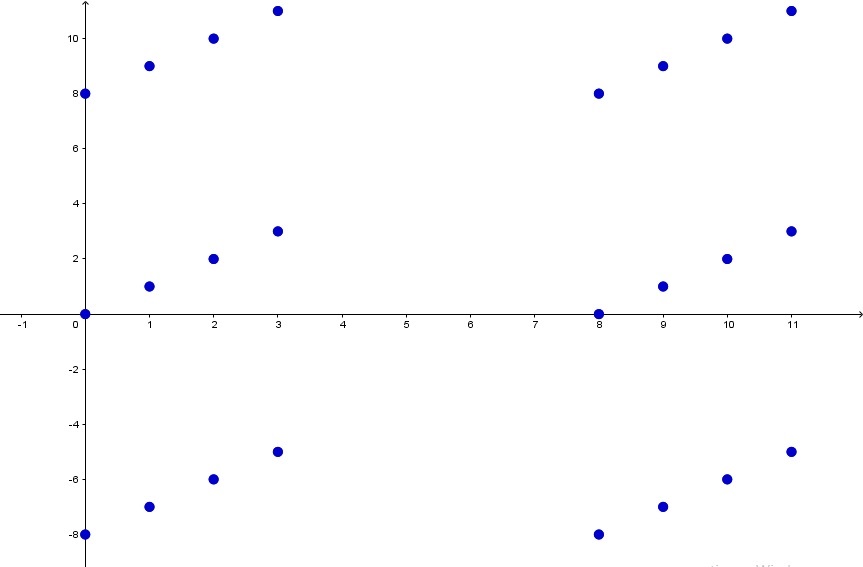}  
\caption{{Some elements of Construction C, with $\mathcal{C}_{1}=\mathcal{C}_{2}=\{(0,0),(1,1)\}$ and $\mathcal{C}_{3}=\{(0,0)\}.$}} \label{p2}
\end{center}
\end{figure}
\end{example}


\section{General geometrically uniform constellations} \label{sec4}

	The next theorem provides a way to construct a geometrically uniform constellation, from which we can derive the geometric uniformity of Construction C for $L=2$ levels.

	
\begin{theorem}\label{generalgu} (\textit{Geometric uniformity of $\Lambda+\mathcal{C}$}) If $\Lambda \subset \mathbb{R}^n$ is a lattice which is symmetric with respect to all coordinate axes, i.e., $(x_1,\dots, x_n) \in \Lambda$ then $(x_1,\dots, x_{i-1}, -x_i, x_{i+1}, \dots, x_n) \in \Lambda$ for all $i$ and $\mathcal{C} \subseteq \mathbb{F}_{2}^{n}$ is a linear binary code, then $\Gamma=\Lambda+\mathcal{C}$ is geometrically uniform.
\end{theorem} 

\begin{proof} Given $x=\lambda_1+c_1 \in \Gamma,$ where $\lambda_1 \in \Lambda$ and $c_1 \in \mathcal{C},$ consider first the linear map $T_{c_1}: \mathbb{R}^{n} \rightarrow \mathbb{R}^{n},$ $T_{c_1}(z)=[T_{c_1}] \cdot z$ ($z$ in the column format), where $[T_{c_{1}}]$ is defined as

\begin{equation} \label{eqT}
[T_{c_1}]=\begin{pmatrix}
	(-1)^{c_{11}} & 0 & \dots & 0 \\
	0 & (-1)^{c_{12}} & \dots & 0  \\
	0 & 0 & \ddots & 0 \\
	0 & 0 & \dots & (-1)^{c_{1n}} 
	\end{pmatrix}_{(n \times n)},
\end{equation}
$c_{1}=(c_{11}, c_{12}, \dots, c_{1n}).$ Observe that $T_{c_1}$ is an isometry and ${T_{c_1}}^{-1}=T_{c_1}.$ 

	The map $F_x: \mathbb{R}^{n} \rightarrow \mathbb{R}^{n},$ $F_x(y)=T_{c_1}(y-x)$ is an isometry and we show next that its restriction $F_x|_{\Gamma}: \Gamma \rightarrow \Gamma$ is also an isometry with $F_x(x)=0.$ 
	
	First, note that for $c_{1}, c_{2} \in \mathcal{C}$ it holds that $T_{c_1}(c_{2}-c_{1})=c_{1} \oplus c_{2}.$ In fact,
	
\begin{equation}
(T_{c_1}(c_{2}-c_{1}))_{i} = \begin{cases}
0, & \text{if} \ (c_{1i},c_{2i})=(0,0) \\
1, & \text{if} \ (c_{1i},c_{2i})=(1,0) \\
1, & \text{if} \ (c_{1i},c_{2i})=(0,1) \\
0, & \text{if} \ (c_{1i},c_{2i})=(1,1)
\end{cases}
\end{equation}
what implies $T_{c_1}(c_2 - c_1)=c_{1} \oplus c_{2}.$ 

	Given $y \in \Gamma= \Lambda+ \mathcal{C},$ $y=\lambda_{2}+c_{2},$
\begin{eqnarray}
F_{x}(y)=T_{c_{1}}(y-x)& = &T_{c_1}(\lambda_2-\lambda_1 +c_2 -c_1) = T_{c_1}(\lambda_2 - \lambda_1) + T_{c_{1}}(c_{2}-c_1)  \nonumber \\
&=& \lambda_3 + (c_{1} \oplus c_{2}) \in \Gamma= \Lambda+ \mathcal{C},
\end{eqnarray}

since $\Lambda$ is axes-symmetric. Therefore, $F_x(\Gamma) \subseteq \Gamma.$

	As $F_x$ is injective, it remains to prove that for any $w= \tilde{\lambda}+\tilde{c} \in \Gamma$ there exists $y \in \Gamma$ such that $w=F_x(y).$ By straightforward calculation we can see that 
\begin{eqnarray}
F_{x}(y)= \tilde{\lambda}+ \tilde{c} & \Rightarrow & T_{c_1}(y-(\lambda_1 + c_1)) = \tilde{\lambda}+ \tilde{c} \Rightarrow T_{c_1}(T_{c_1}(y-(\lambda_1 + c_1))) = T_{c_{1}}(\tilde{\lambda}+ \tilde{c}) \nonumber \\
& \Rightarrow & y = T_{c_1} (\tilde{\lambda}) + \lambda_1 + T_{c_1}(\tilde{c}) + c_1 = T_{c_1} (\tilde{\lambda}) + \lambda_1 + T_{c_1}(\tilde{c} + 0 - c_{1}) \nonumber \\
& \Rightarrow & y = T_{c_1} (\tilde{\lambda}) + \lambda_1 + T_{c_1}(\tilde{c} - c_{1}) = \underbrace{T_{c_1} (\tilde{\lambda}) + \lambda_1}_{ \in \Lambda} + \underbrace{\tilde{c} \oplus c_{1}}_{\in \mathcal{C}} \in \Gamma.
\end{eqnarray}

	To conclude the proof, given any $ x \in \Gamma$ and $w \in \Gamma,$ we can consider the isometry $F_{w}^{-1} \circ F_x$ which takes $x$ to $w.$ 
\end{proof}

\begin{corollary} \label{corogu} (Special geometrically uniform Construction C) If an $L-$level Construction C has just two nonzero linear codes $C_{i}$ and $C_{L},$ for some $1 \leq i \leq L-1,$ then $\Gamma_{C}=2^{i-1}\mathcal{C}_{i}+2^{L-1}\mathcal{C}_{L}+2^{L}\mathbb{Z}^{n}$ is geometrically uniform. 
\end{corollary}

\begin{proof} We can write $\Gamma_{C}=2^{i-1}(\mathcal{C}_{i} + 2^{L-i} (\mathcal{C}_{L} + 2\mathbb{Z}^{n})).$ Since the Construction A lattice $\mathcal{C}_{L} + 2\mathbb{Z}^{n}$ is axes-symmetric and so its scaling by $2^{L-i},$ it follows from Theorem \ref{generalgu} that $\mathcal{C}_{i} + 2^{L-i} (\mathcal{C}_{L-1} + 2\mathbb{Z}^{n}),$ $i=1, \dots, L-1,$ is geometrically uniform, and this fact also applies to the scaled version.
\end{proof}

\begin{remark}\label{remgu} The fact that a $2-$level Construction C is geometrically uniform (Section \ref{subsgu}) is a special case of Corollary \ref{corogu} for $L=2$ and $i=1.$
\end{remark}

	
\section{Construction $C^\star:$ an inter-level coded version of Construction C}

	In this section, we introduce a new method to construct constellations from binary codes, called Construction $C^\star,$ which generalizes the multilevel Construction C.
	
\begin{definition} \label{constrcstar} (\textit{Construction $C^{\star}$}) Let  $\mathcal{C}$ be a length$-nL$ binary code, $\mathcal{C} \subseteq \mathbb{F}_{2}^{nL},$ which we denote by main code. Then Construction $C^{\star}$ is a discrete subset of $\mathbb{R}^{n}$ defined as 
\begin{eqnarray}
 \Gamma_{C^{\star}}& := & \{c_{1}+2c_{2}+ \dots + 2^{L-1}c_{L}+2^{L}z: (c_{1}, c_{2}, \dots, c_{L}) \in \mathcal{C}, \nonumber \\
& & c_{i} \in \mathbb{F}_{2}^{n}, i=1, \dots, L, z \in \mathbb{Z}^{n}\}.  
\end{eqnarray}
\end{definition}

\begin{definition} (\textit{Projection codes}) \label{subcodes} Let $c=(c_1,c_2,...,c_L)$ be a partition of a codeword $c = (c_{11}, \dots, c_{1n},....,c_{L1}, \dots, c_{Ln}) \in \mathcal{C} \subseteq \mathbb{F}_2^n$ into length$-n$ subvectors  $c_i = (c_{i1},....,c_{in}),$  $i=1,\dots,L.$ Then, a projection code $\mathcal{C}_i$ consists of all subvectors $c_{i}$ that appear as we scan through all possible codewords $c \in \mathcal{C}.$ Note that if $\mathcal{C}$ is linear, then every projection code $\mathcal{C}_{i}, i=1, \dots, L$ is also linear.
\end{definition}

\begin{remark}\label{rem2} Construction $C^\star$ is a generalization of Construction C. Specifically, if the main code in $\mathbb{F}_{2}^{nL}$ is given as $\mathcal{C}=\mathcal{C}_{1} \times \mathcal{C}_{2} \times \dots \times \mathcal{C}_{L},$ then Construction $C^{\star}$ becomes Construction C, because the projection codes are independent.
\end{remark}

\begin{definition}(\textit{Associated Construction C}) \label{associated} Given a Construction $C^{\star}$ defined by a binary code $\mathcal{C} \subseteq \mathbb{F}_{2}^{nL},$ we call the associated Construction C the constellation defined as
\begin{equation}
\Gamma_{{C}}= \Gamma_{C} (\Gamma_{C^\star})=\mathcal{C}_{1} + 2 \mathcal{C}_{2} + \dots + 2^{L-1}\mathcal{C}_{L}+2\mathbb{Z}^{n},
\end{equation}
where $\mathcal{C}_{1}, \mathcal{C}_{2}, \dots, \mathcal{C}_{L} \subseteq \mathbb{F}_{2}^{n}$ are the projection codes of $\mathcal{C}$ (Definition \ref{subcodes}).
\end{definition}

\begin{remark} Construction $C^\star$ is always a subset of its associated Construction $C,$ i.e., $\Gamma_{C^\star} \subseteq \Gamma_{C},$ because in general the main code $\mathcal{C}$ restricts the possible combinations of the component codes (unless they are independent as in Remark \ref{rem2}).
\end{remark}

\begin{example}(A nonlattice Construction $C^\star$ and its associated Construction C) \label{ex1} Consider a linear binary code $\mathcal{C}$ with length $nL=4$ ($L=n=2$), where $\mathcal{C}=\{(0,0,0,0),$ $(1,0,0,1),(1,0,1,0),$ $(0,0,1,1)\} \subseteq \mathbb{F}_{2}^{4}.$ 
	Thus, any element $x(c,z) \in \Gamma_{C^{\star}}$ can be written as
\begin{equation}
x(c,z)=\begin{cases}
(0,0)+4z, & \mbox{if} \ {c}_{1}=(0,0) \ \mbox{and} \ {c}_{2}=(0,0) \\
(1,2)+4z, & \mbox{if} \ {c}_{1}=(1,0) \ \mbox{and} \ {c}_{2}=(0,1) \\
(3,0)+4z, & \mbox{if} \ {c}_{1}=(1,0) \ \mbox{and} \ {c}_{2}=(1,0) \\
(2,2)+4z, & \mbox{if} \ {c}_{1}=(0,0) \ \mbox{and} \ {c}_{2}=(1,1), \\
\end{cases}
\end{equation}	
where $(c_1, c_2) \in \mathcal{C}$ and $z \in \mathbb{Z}^{2}.$ Geometrically, the resulting constellation is given by the blue circles represented in Figure \ref{consc}. We can notice that $\Gamma_{C^{\star}}$ is not a lattice. However, if we consider the associated Construction C with codes $\mathcal{C}_{1}=\{(0,0),(1,0)\}$ and $\mathcal{C}_{2}=\{(0,0),(1,1),(0,1),(1,0)\},$ we have a lattice (pink points in Figure \ref{consc}), because $\mathcal{C}_{1}$ and $\mathcal{C}_{2}$ satisfy the condition given by Theorem \ref{thmko}.

\begin{figure}[H]
\begin{center}
		\includegraphics[height=5.1cm]{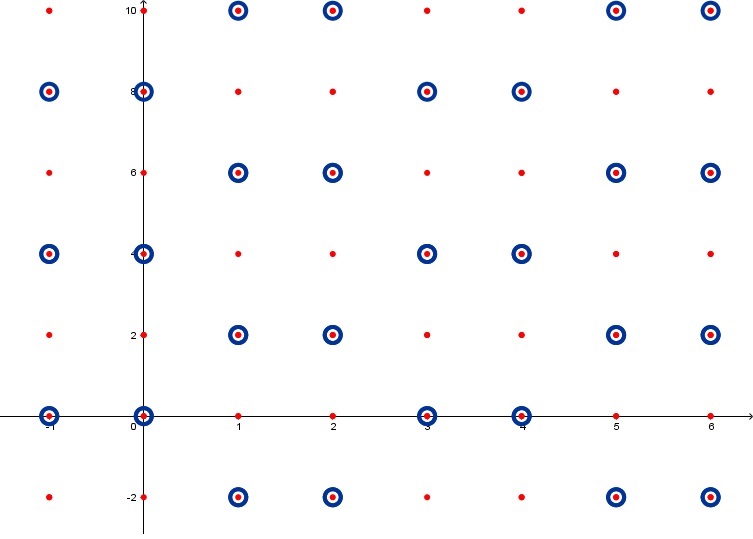}  
\caption{{(Nonlattice) Construction $C^{\star}$ constellation in blue and its associated (lattice) Construction C constellation in red.}}
 \label{consc}
\end{center}
\end{figure} 	
\end{example}			

	The next example presents a case where both Construction $C^\star$ and its associated Construction $C$ are lattices, but they are not equal.
	
\begin{example}(Lattice Constructions $C^\star$ and associated C) \label{ceconjecture} Let a linear binary code $\mathcal{C}=\{(0,0,0,0),$ $(0,0,1,0),(1,0,0,1),(1,0,1,1)\} \subseteq \mathbb{F}_{2}^{4}$ ($nL=4,$ $L=n=2$), so the projection codes are $\mathcal{C}_{1}=\{(0,0),(1,0)\}$ and $\mathcal{C}_{2}=\{(0,0),(1,0),$ $(0,1),(1,1)\}.$ An element $x(c,z) \in \Gamma_{C^{\star}}$ can be described as 
\begin{equation}
x(c,z)=\begin{cases}
(0,0)+4z, & \mbox{if} \ {c}_{1}=(0,0) \ \mbox{and} \ {c}_{2}=(0,0) \\
(1,2)+4z, & \mbox{if} \ {c}_{1}=(1,0) \ \mbox{and} \ {c}_{2}=(0,1) \\
(2,0)+4z, & \mbox{if} \ {c}_{1}=(0,0) \ \mbox{and} \ {c}_{2}=(1,0) \\
(3,2)+4z, & \mbox{if} \ {c}_{1}=(1,0) \ \mbox{and} \ {c}_{2}=(1,1), \\
\end{cases}
\end{equation}
with $z \in \mathbb{Z}^{2}.$ This construction is represented by black circles in Figure \ref{fig10}. Note that $\Gamma_{C^{\star}}$ is a lattice and $\mathcal{C} \neq \mathcal{C}_{1} \times \mathcal{C}_{2},$ which implies that $\Gamma_{C^{\star}} \neq \Gamma_{C}.$ Nevertheless, the associated Construction C is also a lattice (Figure \ref{fig10}).

	One can clearly observe the advantage of $\Gamma_{C^\star}$ over the associated $\Gamma_{{C}}$ in this case, because the packing densities are, respectively $\Delta_{\Gamma_{C^\star}} = \frac{\Pi}{4} \approx 0.7853$ and  $\Delta_{\Gamma_{{C}}} = \frac{\Pi}{8} \approx 0.3926.$ 

\begin{figure}[H]
\begin{center}
		\includegraphics[height=5.1cm]{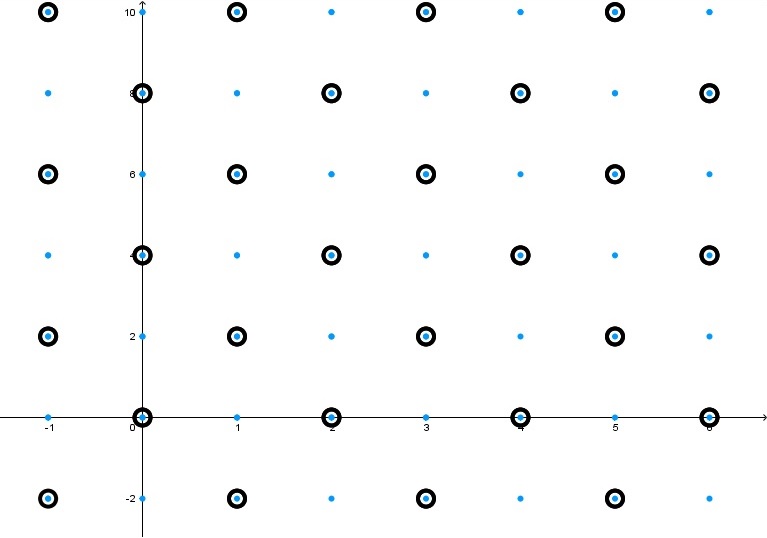}  
\caption{{(Lattice) Construction $C^{\star}$ constellation in black and its associated (lattice) Construction C constellation in blue.}}
 \label{fig10}
\end{center}
\end{figure} 	
\end{example}

	We can also describe the densest lattice in dimension $24$, the Leech lattice $\Lambda_{24},$ as a Construction $C^{\star}$ with $3$ levels. 

\begin{example} (Leech lattice from Construction $C^\star$) \label{exleech} Based on the construction given by Conway and Sloane \cite{conwaysloane} (pp. 131-132) and Amrani et al \cite{amrani}, we start by considering the following three linear binary codes:
\vspace{0.1cm}
\begin{itemize}
\item $\mathcal{C}_{1}=\{(0,\dots, 0), (1, \dots, 1)\} \subseteq \mathbb{F}_{2}^{24};$
\vspace{0.2cm}
\item $\mathcal{C}_{2}$ as a Golay code $\mathcal{C}_{24} \subset \mathbb{F}_{2}^{24}$ achieved by adding a parity bit to the original $[23,12,7]-$binary Golay code $\mathcal{C}_{23},$ which consists in a quadratic residue code of length 23;
\vspace{0.2cm}
\item $\mathcal{C}_{3} = \tilde{\mathcal{C}}_{3} \cup \overline{\mathcal{C}}_{3} = \mathbb{F}_{2}^{24},$ where $\tilde{\mathcal{C}}_{3}=\{(x_{1}, \dots, x_{24}) \in \mathbb{F}_{2}^{24}:  \sum_{i=1}^{24} x_{1} \equiv 0 \mod 2\}$ and $\overline{\mathcal{C}}_{3}=\{(y_{1}, \dots, y_{24}) \in \mathbb{F}_{2}^{24}: \sum_{i=1}^{24} y_{1} \equiv 1 \mod 2\}.$
\end{itemize}

	Observe that $\mathcal{C}_{1}, \mathcal{C}_{2}$ and $\mathcal{C}_{3}$ are linear codes. Consider a code $\mathcal{C} \subseteq \mathbb{F}_{2}^{72}$ whose codewords are described in one of two possible ways:
\begin{eqnarray}
\mathcal{C}=\{ (\underbrace{0,\dots, 0}_{\in \mathcal{C}_1}, \underbrace{a_{1}, \dots, a_{24}}_{\in \mathcal{C}_{2}}, \underbrace{x_{1}, \dots, x_{24}}_{\in \tilde{\mathcal{C}}_{3}}), (\underbrace{1, \dots, 1}_{\in \mathcal{C}_1}, \underbrace{a_{1}, \dots, a_{24}}_{\in \mathcal{C}_{2}}, \underbrace{y_{1}, \dots, y_{24}}_{\in \overline{\mathcal{C}}_{3}})\}.
\end{eqnarray}

	Thus, we can define the Leech lattice $\Lambda_{24}$ as the $3-$level Construction $C^{\star}$ given by
\begin{equation}
\Lambda_{24}=\Gamma_{C^{\star}}=\{c_{1}+2c_{2}+4c_{3}+8z: (c_{1}, c_{2}, c_{3}) \in \mathcal{C}, z \in \mathbb{Z}^{24}\}.
\end{equation}
Observe that $\Gamma_{C^{\star}} \neq \Gamma_{C}$ and
in this case, the associated Construction C has packing density $\Delta_{\Gamma_{C}} \approx 0.00012 < 0.001929 \approx \Delta_{\Gamma_{C^\star}},$ which is the packing density of $\Lambda_{24},$ the best known packing density in dimension $24$ even for nonlattices \cite{cohnetal}. 
\end{example}

\section{Geometric uniformity and latticeness of Construction $\mathcal{C}^\star$} \label{gucstar}

	We verified previously (Section \ref{subsgu} and Remark \ref{remgu}) that a $2-$level Construction C, $\Gamma_{C}=\mathcal{C}_{1}+2\mathcal{C}_{2}+\mathbb{Z}^{n},$
where $\mathcal{C}_{1}, \mathcal{C}_{2} \subseteq \mathbb{F}_{2}^{n}$ are linear codes, is always geometrically uniform even when it is not a lattice. Another question that emerges is: is a $2-$level Construction $C^{\star}$ also geometrically uniform? As we show in a sequel, the answer is affirmative.

\begin{theorem} \label{gucstar} (Geometric uniformity of a $2-$level Construction $C^{\star}$) For any binary linear code $\mathcal{C} \subseteq \mathbb{F}_{2}^{2n},$ the $2-$level Construction $C^\star$ defined as $\Gamma_{C^{\star}} = \{c_{1}+2c_{2}+4z: (c_{1}, c_{2}) \in \mathcal{C}, z \in \mathbb{Z}^{n}\}$ is geometrically uniform.
\end{theorem}

\begin{proof} Assuming the same isometry $T_{c_1}$ given by Equation \eqref{eqT}, for a fixed $x=c_1+2c_2+4z \in \Gamma_{C^\star},$ with $(c_1, c_2) \in \mathcal{C}$ and given $y=\tilde{c}_1+2\tilde{c}_2+4\tilde{z} \in \Gamma_{C^\star},$ $(\tilde{c}_1, \tilde{c}_2) \in \mathcal{C},$ it is true that $T_{c_1}(y-x)=((\tilde{c}_1-{c}_1)  \bmod 2) + 2 ((\tilde{c}_2-{c}_2)  \bmod 2)+4z',$ where $z' \in \mathbb{Z}$ is suitably chosen according to the value of the difference in each coordinate. Clearly, by the linearity of the code, $((\tilde{c}_1-{c}_1)  \bmod 2, (\tilde{c}_2-{c}_2)  \bmod 2) \in \mathcal{C},$ then $T_{c_1}(y-x) \in \Gamma_{C^\star}.$ To prove that  $T_{c_1}(\Gamma_{C^\star} -x) = \Gamma_{C^\star},$ we use arguments as the ones in Theorem \ref{generalgu}. Therefore, we can conclude that for $L=2,$ $\Gamma_{C^\star}$ is geometrically uniform.

\end{proof}

	As we have seen in Example \ref{cex2}, Construction C is not geometrically uniform for general $L \geq 3.$ The same holds for Construction $C^\star$, because if we consider $\mathcal{C}$ as the Cartesian product $\mathcal{C}= \mathcal{C}_{1} \times \mathcal{C}_{2} \times  \dots \times \mathcal{C}_{L},$ then $\Gamma_{C}= \Gamma_{C^\star}$ and therefore, it is also not geometrically uniform in general for $L \geq 3.$
	
	The work of Kositwattanarerk and Oggier \cite{kositoggier} motivated our search for the latticeness of Construction $C^\star.$ Note that in \cite{kositoggier}, the requirements for Construction C to be a lattice are imposed by a comparison between Constructions C and D. This is not possible for Construction $C^\star,$ which requires a different strategy. In the upcoming discussion, we present a simple sufficient condition (Theorem \ref{coro4}) followed by a necessary and sufficient condition (Theorem \ref{thmcomplete}) for $\Gamma_{C^{\star}}$ to be a lattice.


	We define now the antiprojection set, which is an important notion to state the lattice conditions of Construction $C^\star.$

\begin{definition} (\textit{Antiprojection}) The antiprojection (inverse image of a projection) $\mathcal
{S}_{i}(c_1,\dots, c_{i-1},$ $c_{i+1}, \dots, c_{L})$ consists of all vectors $c_{i} \in \mathcal{C}_{i},$ $i=1,\dots, L$ that appear as we scan through all possible codewords $c \in \mathcal{C},$ while keeping $c_{1}, \dots, c_{i-1}, c_{i+1}, \dots, c_{L}$ fixed:
\begin{eqnarray}
&\mathcal{S}_{i}(c_1,...,c_{i-1}, c_{i+1},...,c_{L})=\{c_{i} \in \mathcal{C}_{i}: (c_{1}, \dots, \underbrace{c_{i}}_{\text{\makebox[0pt]{i-th posititon} }}, \dots, c_{L}) \in \mathcal{C}\}.
\end{eqnarray}
\end{definition}

\begin{example}(Antiprojection set) In Example \ref{ceconjecture}, we can define the antiprojection $\mathcal{S}_{2}(c_{1})= \{c_{2} \in \mathcal{C}_{2}: (c_{1},c_{2}) \in \mathcal{C}\}.$ For $c_{1}=(0,0) \in \mathcal{C}_{1}$ we have $\mathcal{S}_{2}(c_{1})=\{(0,0),(1,0)\}$ and for $c_{1}=(1,0) \in \mathcal{C}_{1},$ $\mathcal{S}_{2}(c_{1})=\{(0,1),(1,1)\}.$ 
\end{example}

	We now can state a sufficient condition for the latticeness of Construction $C^\star.$

\begin{theorem}  (\textit{A sufficient lattice condition for $\Gamma_{C^\star}$})  \label{coro4} If $\mathcal{C} \subseteq \mathbb{F}_{2}^{nL}$ is a linear binary code with projection codes $\mathcal{C}_{1},\mathcal{C}_{2}, \dots, \mathcal{C}_{L}$ such that $\mathcal{C}_{1} \subseteq \mathcal{S}_{2}(0,\dots,0) \subseteq \mathcal{C}_{2} \subseteq \dots \subseteq \mathcal{C}_{L-1} \subseteq \mathcal{S}_{L}(0, \dots,0) \subseteq \mathcal{C}_{L} \subseteq \mathbb{F}_{2}^{n}$ and where $\mathcal{C}_{i-1} \subseteq \mathcal{S}_{i}(0,\dots,0)$ is closed under the Schur product for $i=2,\dots, L,$ then $\Gamma_{C^{\star}}$ is a lattice.
\end{theorem}

	The proof will be presented in what follows after the statement of an auxiliary result. One can observe that the mathematical intuition behind the conditions for $\Gamma_{C^{\star}}$ to be a lattice lies in the fact that, since $a + b = a \oplus b + 2 (a \ast b)$  for $a,b \in \mathbb{F}_{2}^{n},$ when adding two points in $\Gamma_C$ or $\Gamma_{C^\star},$ each level $i \geq 2$ has the form of $c_{i} \oplus \tilde{c}_i \oplus carry_{(i-1)},$ where $carry_{(i-1)}$ is the ``carry" term vector of the $2(a \ast b)$ coordinates from the addition in the lower level. Since the projection code $\mathcal{C}_i$ is linear, $c_i \oplus \tilde{c}_i$ is a codeword in the $i-$th level. Hence, closeness of $\Gamma_{C^\star}$ under addition amounts to the fact that $carry_{(i-1)}$ is also a codeword in $\mathcal{C}_i,$ for $i=2,3, \dots,L.$ The next result presents a detailed description of these ``carry" terms. 

\begin{lemma} (\textit{Sum in $\Gamma_{C^{\star}}$}) \label{lemmasum} Let $\mathcal{C} \subseteq \mathbb{F}_{2}^{nL}$ be a binary linear code. If $x,y \in \Gamma_{C^{\star}}$ are such that
\begin{eqnarray}
x&=&c_{1}+2c_{2}+\dots+2^{L-1}c_{L}+2^{L}z \label{eqx} \\
y&=&\tilde{c}_{1}+2\tilde{c}_{2}+\dots+2^{L-1}\tilde{c}_{L}+2^{L}\tilde{z}, \label{eqy}
\end{eqnarray}
with $(c_{1}, c_{2}, \dots, c_{L}), (\tilde{c}_{1}, \tilde{c}_{2}, \dots, \tilde{c}_{L}) \in \mathcal{C}$ and $z, \tilde{z} \in \mathbb{Z}^{n},$ then
{\begin{align} \label{formulasum}
x+y &= c_{1}\oplus \tilde{c}_{1} + 2(s_{1} \oplus (c_{2}\oplus \tilde{c}_{2})) + \dots +2^{L-1}(s_{L-1} \oplus (c_{L}\oplus \tilde{c}_{L}))+ 2^{L}(s^{\star}_{L}+z+\tilde{z}),
\end{align}}
where $s_i \in \mathbb{F}_{2}^{n}$ is the ``carry" from level $i$ to level $i+1,$ given by 
\begin{eqnarray} \label{si}
& s_{i}= (c_{i} \ast \tilde{c}_{i}) \oplus r_{i}^{(1)} \oplus r_{i}^{(2)} \oplus \dots \oplus r_{i}^{(i-1)} = (c_{i} \ast \tilde{c}_{i}) \bigoplus\limits_{j=1}^{i-1} r_{i}^{(j)}, \nonumber \\
& r_{i}^{(1)}=(c_{i} \oplus \tilde{c}_{i}) \ast (c_{i-1} \ast \tilde{c}_{i-1}), \ \ r_{i}^{(j)}=r_{i}^{(j-1)} \ast r_{i-1}^{(j-1)}, \nonumber \\
& 2 \leq j \leq i-1, i=2, \dots, L-1
\end{eqnarray}
$s_{0}=(0, \dots, 0),$ $s_1=c_{1} \ast \tilde{c}_{1}$ and the formula for $s^{\star}_{i}$ is the same for $s_{i}$ but with real sum instead of modulo-2 sum. 
\end{lemma}

\begin{proof} Through induction in the number $L$ of levels:

\noindent \underline{Base case:} For $L=1$ level, $\mathcal{C} \subseteq \mathbb{F}_{2}^{n}$ has only one subcode $\mathcal{C}_{1}.$ Consider $x, y \in \Gamma_{C^\star}$ such that $x=c_{1}+2z$ and $y=\tilde{c}_{1}+2\tilde{z}.$ Then
\begin{eqnarray}
x+y &=& c_{1} + \tilde{c}_{1} + 2(z+ \tilde{z}) = c_{1} \oplus \tilde{c}_{1} + 2(\underbrace{c_{1} \ast \tilde{c}_{1}}_{s_{1} \in \mathbb{Z}^{n}} + z+ \tilde{z})
\end{eqnarray}
and the result is valid.\\

\noindent \underline{Induction step:} Assume that the formula in Equation (\ref{formulasum}) is valid for $L=k-1,$ where the main code $\tilde{\mathcal{C}} \in \mathbb{F}_{2}^{n(k-1)}$ has subcodes $\mathcal{C}_{1}, \dots, \mathcal{C}_{k-1} \in \mathbb{F}_{2}^{n}.$ Therefore, our induction hypothesis affirms that for $x,y \in \Gamma_{C^\star}$ such that
\begin{eqnarray}
x&=&c_{1}+2c_{2}+\dots+2^{k-2}c_{k-1}+2^{k-1}z \label{eqx} \\
y&=&\tilde{c}_{1}+2\tilde{c}_{2}+\dots+2^{k-2}\tilde{c}_{k-1}+2^{k-1}\tilde{z}, \label{eqy}
\end{eqnarray}
with $z, \tilde{z} \in \mathbb{Z}^{n},$ is true that  
\begin{eqnarray}
x+y &=& c_{1}\oplus \tilde{c}_{1} + 2(s_{1} \oplus (c_{2}\oplus \tilde{c}_{2})) + \dots + 2^{k-2}(s_{k-2} \oplus (c_{k-1}\oplus \tilde{c}_{k-1})) \nonumber \\
&+&2^{k-1}(s^{\star}_{k-1}+z+\tilde{z}),
\end{eqnarray}
where $s^{\star}_{k-1}$ is $s_{k-1}$ and $s_{i}, i=1, \dots, L$ as in Equation (\ref{si}).

	We aim to prove that the formula presented in Equation (\ref{formulasum}) is also satisfied for $L=k.$ So, consider the main code $\mathcal{C} \in \mathbb{F}_{2}^{nk}$ with subcodes $\mathcal{C}_{1}, \dots, \mathcal{C}_{k-1}, \mathcal{C}_{k} \in \mathbb{F}_{2}^{n}.$ Suppose $\overline{x},\overline{y} \in \Gamma_{C^\star}$ such that
\begin{eqnarray}
\overline{x}&=&c_{1}+2c_{2}+\dots+2^{k-2}c_{k-1}+2^{k-1}c_{k}+2^{k}z \\
\overline{y}&=&\tilde{c}_{1}+2\tilde{c}_{2}+\dots+2^{k-2}\tilde{c}_{k-1}+2^{k-1}\tilde{c}_{k}+2^{k}\tilde{z}.
\end{eqnarray} 
So we can write, applying the induction hypothesis
\begin{eqnarray}
\overline{x}+\overline{y} &=& c_{1}\oplus \tilde{c}_{1} + 2(s_{1} \oplus (c_{2}\oplus \tilde{c}_{2})) + \dots + 2^{k-2}(s_{k-2} \oplus (c_{k-1}\oplus \tilde{c}_{k-1})) + \nonumber \\
& &2^{k-1}(s^{\star}_{k-1} + c_{k}+ \tilde{c}_{k})+2^{k}(z+\tilde{z}),
\end{eqnarray}
where $s^{\star}_{k-1}$ is $s_{k-1}$ with the real sum instead of modulo$-2$ sum. By doing all the decompositions to change the real sum $s^{\star}_{k-1} + c_{k}+ \tilde{c}_{k}$ to $s_{k-1} \oplus c_{k} \oplus \tilde{c}_{k}$ we have
\begin{eqnarray}
\overline{x}+\overline{y} &=& c_{1}\oplus \tilde{c}_{1} + 2(s_{1} \oplus (c_{2}\oplus \tilde{c}_{2})) + \dots + 2^{k-2}(s_{k-2} \oplus (c_{k-1}\oplus \tilde{c}_{k-1})) + \nonumber \\
& &2^{k-1}(s_{k-1} \oplus (c_{k} \oplus \tilde{c}_{k}))+2^{k}(\underbrace{(c_{k} \ast \tilde{c}_{k}) + r_{k}^{(1)}+r_{k}^{(2)} + \dots + r_{k}^{(k-1)}}_{s^{\star}_{k}}+ z+\tilde{z}).
\end{eqnarray}
This formula is exactly as we expected and it concludes the proof.
\end{proof} 

	At this point, we are ready to present the proof of Theorem \ref{coro4}.

\begin{proof} For any $x, y \in \Gamma_{C^{\star}},$ written as in Equations (\ref{eqx}) and (\ref{eqy}), we have $x+y$ as given in Lemma \ref{lemmasum} (Equations (\ref{formulasum}) and (\ref{si})) and we need to verify if $x+y \in \Gamma_{C^{\star}}.$

Clearly $x+y \in \mathcal{C}_{1}+2\mathcal{C}_{2}+\dots+2^{L-1}\mathcal{C}_{L}+2^{L}\mathbb{Z}^{n}.$ It remains to demonstrate that $(c_{1}\oplus \tilde{c}_{1},  s_{1} \oplus c_{2}\oplus \tilde{c}_{2}, \dots, s_{L-1} \oplus c_{L} \oplus \tilde{c}_{L}) \in \mathcal{C},$ where $s_1, \dots, s_{L-1}$ are the ``carry'' terms defined in Equation \eqref{si}. Using the fact that $\mathcal{C}_{i-1} \subseteq \mathcal{S}_{i}(0, \dots, 0)$ for all $i=2, \dots, L$ are closed under the Schur product, it holds that
\begin{align}
(c_{1}\oplus \tilde{c}_{1},  s_{1} \oplus c_{2}\oplus \tilde{c}_{2}, \dots, s_{L-1} \oplus c_{L}\oplus \tilde{c}_{L}) & = \underbrace{(c_{1}\oplus \tilde{c}_{1},c_{2}\oplus \tilde{c}_{2}, \dots,  c_{L}\oplus \tilde{c}_{L})}_{\in \mathcal{C}} \oplus  \nonumber \\
& \oplus \underbrace{(0,s_{1}, \dots, 0)}_{\in \mathcal{C}} \oplus \dots \oplus  \underbrace{(0,\dots, 0, s_{L-1})}_{\in \mathcal{C}} \in \mathcal{C}.
\end{align}
Observe that any $nL-$tuple $(0, \dots, s_{i-1}, \dots, 0) \in \mathcal{C}$ because by hypothesis, the chain $\mathcal{S}_{i}(0, \dots, 0)$ closes $\mathcal{C}_{i-1}$ under Schur product, hence $S_{i}(0, \dots, 0)$ contains $(c_{i-1}*\tilde{c}_{i-1}), r_{i-1}^{(1)},....,r_{i-1}^{(i-2)}$ which is sufficient to guarantee that $s_{i-1} \in \mathcal{S}_{i}(0,\dots,0)$ so $(0, \dots, s_{i-1}, \dots, 0) \in \mathcal{C},$ for all $i=2, \dots, L-1.$ Also, if $x \in \Gamma_{C^\star},$ it is also valid that $-x \in \Gamma_{C^\star}.$
\end{proof}

	While $\mathcal{S}_{i}(0,\dots,0) \subseteq \mathcal{C}_{i}, \ i=1, \dots, L$ by construction, note that the assumption that $\mathcal{C}_{i} \subseteq \mathcal{S}_{i+1}(0,\dots,0),$ $i=1, \dots, L-1$ in Theorem \ref{coro4} is not always satisfied by a general Construction $C^{\star},$ even when Construction $C^\star$ is a lattice (Example \ref{exnonesting}). On the other hand, the conditions provided by Theorem \ref{coro4} are easier to verify, as it can be seen, for example, in the case of the Leech lattice.
	
\begin{example}(Latticeness of the Leech lattice via Theorem \ref{coro4}) We want to check whether the proposed codes $\mathcal{C}_{1}, \mathcal{C}_{2}$ and $\mathcal{C}_{3}$ in Example \ref{exleech} satisfy the conditions stated by Theorem \ref{coro4}.

	Initially, one can observe that $\mathcal{S}_{2}(0,\dots,0)= \mathcal{C}_{2}$ and  $\mathcal{S}_{3}(0,\dots,0)= \tilde{\mathcal{C}}_{3}=\{(x_{1}, \dots, x_{24}) \in \mathbb{F}_{2}^{24}:  \sum_{i=1}^{24} x_{1} \equiv 0 \bmod 2\}.$ Hence, we need to verify that $\mathcal{C}_{1} \subseteq \mathcal{S}_{2}(0,\dots,0) \subseteq \mathcal{C}_{2}  \subseteq  \mathcal{S}_{3}(0,\dots,0) \subseteq \mathcal{C}_{3}$ and that $\mathcal{S}_{i}(0,\dots,0)$ closes $\mathcal{C}_{i-1}$ under Schur product for $i=2,3.$ 
Indeed $\mathcal{C}_{1} \subseteq \mathcal{S}_{2}(0,\dots,0) = \mathcal{C}_{2},$ since $(0,\dots, 0) \in \mathcal{C}_{2}$ and if we consider the parity check matrix $H \in \mathbb{F}_{2}^{12 \times 24}$ of the $[24,12,8]-$Golay code denoted previously as $\mathcal{C}_2,$
\begin{equation}
H=\begin{pmatrix}
B_{12 \times 12} & \mid & I_{12 \times 12}
\end{pmatrix},
\end{equation}
where 
\begin{equation}
B_{12 \times 12}=\left(
\begin{array}{cccccccccccc}
 1 & 1 & 0 & 1 & 1 & 1 & 0 & 0 & 0 & 1 & 0 & 1  \\
 1 & 0 & 1 & 1 & 1 & 0 & 0 & 0 & 1 & 0 & 1 & 1  \\
 0 & 1 & 1 & 1 & 0 & 0 & 0 & 1 & 0 & 1 & 1 & 1 \\
 1 & 1 & 1 & 0 & 0 & 0 & 1 & 0 & 1 & 1 & 0 & 1  \\
 1 & 1 & 0 & 0 & 0 & 1 & 0 & 1 & 1 & 0 & 1 & 1  \\
 1 & 0 & 0 & 0 & 1 & 0 & 1 & 1 & 0 & 1 & 1 & 1  \\
 0 & 0 & 0 & 1 & 0 & 1 & 1 & 0 & 1 & 1 & 1 & 1  \\
 0 & 0 & 1 & 0 & 1 & 1 & 0 & 1 & 1 & 1 & 0 & 1  \\
 0 & 1 & 0 & 1 & 1 & 0 & 1 & 1 & 1 & 0 & 0 & 1  \\
 1 & 0 & 1 & 1 & 0 & 1 & 1 & 1 & 0 & 0 & 0 & 1  \\
 0 & 1 & 1 & 0 & 1 & 1 & 1 & 0 & 0 & 0 & 1 & 1 \\
 1 & 1 & 1 & 1 & 1 & 1 & 1 & 1 & 1 & 1 & 1 & 0  \\
\end{array}
\right)
\end{equation}
it is straightforward to see that $H \cdot (1, \dots, 1)^{T} = 0 \in \mathbb{F}_{2}^{12},$ so $(1, \dots, 1) \in \mathcal{C}_{2}$ which implies that $\mathcal{C}_{1} \subseteq \mathcal{S}_{2}(0,\dots,0).$ Moreover, an element $c_2 \in \mathcal{C}_{2}$ can be written as $c_2=G \cdot a,$ where $G=\begin{pmatrix}
I_{12 \times 12} \\
\hline 
B_{12 \times 12} 
\end{pmatrix}$ is the generator matrix of the Golay code and $a \in \mathbb{F}_{2}^{12}.$ Thus, when we sum all the coordinates of the resulting vector $c_{2} = G \cdot a$ we have $8 a_1 +8 a_{2} + 8 a_{3} + 8 a_{4} + 8a_{5} + 8a_{6} + 8 a_{7} + 8 a_{8} + 8 a_{9}+8a_{10} +8a_{11} + 12a_{12} \equiv 0 \bmod 2 \Rightarrow c_{2} \in \tilde{\mathcal{C}}_{3} = \mathcal{S}_{3}(0,\dots,0).$ Hence,
\begin{equation}
\mathcal{C}_{1} \subseteq \mathcal{S}_{2}(0,\dots,0) \subseteq \mathcal{C}_{2}  \subseteq  \mathcal{S}_{3}(0,\dots,0) \subseteq \mathcal{C}_{3}.
\end{equation}

	We still need to prove that 
\begin{itemize}
\item $\mathcal{S}_{2}(0,\dots,0)$ closes $\mathcal{C}_{1}$ under Schur product and this is clearly true because the Schur product of all elements in $\mathcal{C}_{1}$ belong to $\mathcal{S}_{2}(0,\dots,0).$

\item $\mathcal{S}_{3}(0,\dots,0)$ closes $\mathcal{C}_{2}$ under Schur product: if we consider $c_{2}= G \cdot a \in \mathcal{C}_{2}$ and $\tilde{c}_{2}=G.\tilde{a} \in \mathcal{C}_{2},$ it can be verified (using the software Mathematica \cite{wolfram}) that the sum of all coordinates of the Schur product $c_{2} \ast \tilde{c}_{2}$ sum zero modulo $2,$ i.e., $c_{2} \ast \tilde{c}_{2} \in \mathcal{S}_{3}(0,\dots,0)= \tilde{C}_{3}.$ 
\end{itemize}
\end{example}
	
	We establish next a necessary and sufficient condition regarding the latticeness of $\Gamma_{C^\star}.$
	
\begin{theorem} (\textit{Lattice condition for $\Gamma_{C^\star}$}) \label{thmcomplete} Let $\mathcal{C} \subseteq \mathbb{F}_{2}^{nL}$ be a linear binary code that generates $\Gamma_{C^\star}$ and let the set $\mathcal{S}=\{(0,s_{1}, \dots, s_{L-1})\} \subseteq \mathbb{F}_{2}^{nL}$ defined for all pairs $c, \tilde{c} \in \mathcal{C}$ (including the case $c=\tilde{c}),$ where $s_i,$ $i=1,\dots, L-1$ are defined as in Equation \eqref{si}.
Then, the constellation $\Gamma_{C^{\star}}$ is a lattice if and only if $\mathcal{S} \subseteq \mathcal{C}.$
\end{theorem}

\begin{proof} $(\Rightarrow)$ Assume $\Gamma_{C^\star}$ to be lattice. This implies that if $x,y \in \Gamma_{C^\star}$ then $x+y \in \Gamma_{C^\star}.$ From the notation and result from Lemma \ref{lemmasum}, more specifically Equations (\ref{formulasum}), (\ref{si}), (\ref{eqx}) and (\ref{eqy}), it means that 
\begin{equation}
(c_{1}\oplus \tilde{c}_{1}, s_{1} \oplus (c_{2}\oplus \tilde{c}_{2}), \dots,s_{L-1} \oplus (c_{L}\oplus \tilde{c}_{L})) \in \mathcal{C}.
\end{equation}

	We can write this $L-$tuple as
\begin{eqnarray}
& & \underbrace{(c_{1}\oplus \tilde{c}_{1}, s_{1} \oplus (c_{2}\oplus \tilde{c}_{2}), \dots,s_{L-1} \oplus (c_{L}\oplus \tilde{c}_{L}))}_{\in \mathcal{C}}  = \nonumber \\ 
& & \underbrace{(c_{1}\oplus \tilde{c}_{1}, c_{2}\oplus \tilde{c}_{2}, \dots, c_{L}\oplus \tilde{c}_{L})}_{\in \mathcal{C}, \ \text{by linearity of} \ \mathcal{C}} \oplus (0, s_{1}, \dots, s_{L-1}) \Rightarrow (0, s_{1}, \dots, s_{L-1}) \in \mathcal{C}, \label{eqshurp}
\end{eqnarray}
which is the same as saying that for all $x,y \in \Gamma_{C^{\star}},$ $\mathcal{S} \subseteq \mathcal{C}.$ 

\noindent $(\Leftarrow)$ The converse is immediate, because given $x,y \in \Gamma_{C^\star}$ as in Equations (\ref{eqx}) and (\ref{eqy}), from the fact that $\mathcal{C}$ is linear and $\mathcal{S} \subseteq \mathcal{C},$ it is valid that
\begin{eqnarray}
& &(c_{1}\oplus \tilde{c}_{1}, c_{2}\oplus \tilde{c}_{2}, \dots, c_{L}\oplus \tilde{c}_{L}) \oplus (0, s_{1}, \dots, s_{L-1}) \in \mathcal{C} \nonumber \\
&\Rightarrow & (c_{1}\oplus \tilde{c}_{1}, s_{1} \oplus (c_{2}\oplus \tilde{c}_{2}), \dots,s_{L-1} \oplus (c_{L}\oplus \tilde{c}_{L})) \in \mathcal{C}
\end{eqnarray}
and $x+y \in \Gamma_{C^{\star}}.$ We still need to prove that $-x \in \Gamma_{C^\star}.$ It is true that for $x \in \Gamma_{C^\star},$ $x+x \in \Gamma_{C^\star}.$ 
 If we do this sum recursively, i.e., $\underbrace{x+x+x+ \dots + x}_{2^{L} \ \text{times}} = 2^{L}j$, for a suitably $j \in \mathbb{Z}^{n}.$ So, if we consider $y= \underbrace{x+x + \dots + x}_{2^{L}-1 \ \text{times}} + 2^{L}(-j) \in \Gamma_{C^\star},$ because it is a sum of elements in $\Gamma_{C^\star}$ for a convenient $-j \in \mathbb{Z}^{n}$ and it follows that $x+y=0 \in \mathbb{R}^{n}$ and $y=-x.$ 
\end{proof}


\begin{remark} \label{remarkc} Note that if $\mathcal{C}=\mathcal{C}_{1} \times \mathcal{C}_{2} \times \dots \times \mathcal{C}_{L},$ i.e., $\Gamma_{C^\star}=\Gamma_{C},$ then Theorem \ref{thmcomplete} specializes to the Kositwattanarerk and Oggier \cite{kositoggier} condition for the latticeness of Construction C. That is, $\mathcal{S} \subseteq \mathcal{C}$ is equivalent to the condition $\mathcal{C}_{1} \subseteq \mathcal{C}_{2} \subseteq \dots \subseteq \mathcal{C}_{L}$ and the chain being closed under Schur product. Indeed,

\begin{itemize}
\item[i)] $\mathcal{S} \subseteq \mathcal{C} \Rightarrow \mathcal{C}_{1} \subseteq \mathcal{C}_{2} \subseteq \dots \subseteq \mathcal{C}_{L}$ and the chain is closed under Schur product: we know that $\mathcal{S} \subseteq \mathcal{C}$ for any pair $c, \tilde{c}$ of codewords, so we take in particular $\tilde{c}=c$ and it follows that $\mathcal{C}_{1} \subseteq \mathcal{C}_{2} \subseteq \dots \subseteq \mathcal{C}_{L}.$ The fact that $\mathcal{C}=\mathcal{C}_{1} \times \mathcal{C}_{2} \times \dots \times \mathcal{C}_{L}$ allows us to guarantee that the element $(0, c_{1} \ast \tilde{c}_{1}, c_{2} \ast \tilde{c}_{2}, \dots, c_{L-1} \ast \tilde{c}_{L-1}) \in \mathcal{S} \subseteq \mathcal{C}$ and then the above chain will being closed under Schur product.

\item[ii)] $\mathcal{C}_{1} \subseteq \mathcal{C}_{2} \subseteq \dots \subseteq \mathcal{C}_{L}$ and the chain is closed under Schur product $\Rightarrow \mathcal{S} \subseteq \mathcal{C}:$ consider an element $(0, s_{1}, s_{2}, \dots, s_{L-1}) \in \mathcal{S},$ we want to prove that this element is also in $\mathcal{C}$ and to do that it is enough to prove that $s_{1} \in \mathcal{C}_{2}, s_{2} \in \mathcal{C}_{3} \dots$ $s_{L-1} \in \mathcal{C}_{L}.$ Indeed, due to the chain being closed under Schur product,
\begin{eqnarray}
s_{1} & = & c_{1} \ast \tilde{c}_{1} \in \mathcal{C}_{2} \\ \label{eqs1}
s_{2} & = &\underbrace{((c_{1} \ast \tilde{c}_{1}) \ast (c_{2} \oplus \tilde{c}_{2}))}_{\in \mathcal{C}_{3}} \oplus \underbrace{(c_{2} \ast \tilde{c}_{2})}_{\in \mathcal{C}_{3}} \in \mathcal{C}_{3} \\ \label{eqs2}
s_{3} & = & \underbrace{((c_{3} \nonumber \oplus \tilde{c}_{3}) \ast (c_{2} \ast \tilde{c}_{2})) \ast (c_{2} \oplus \tilde{c}_{2} \ast (c_1 \ast \tilde{c}_{1}))}_{\in \mathcal{C}_{4}} \nonumber \\
&  & \oplus \ \underbrace{((c_{3} \oplus \tilde{c}_{3}) \ast (c_{2} \ast \tilde{c}_{2}))}_{\in \mathcal{C}_{4}} \oplus \underbrace{(c_{3} \ast \tilde{c}_{3})}_{\in \mathcal{C}_{4}} \in \mathcal{C}_{4} \\ \label{eqs3}
& \vdots & \nonumber
%
\end{eqnarray}
and proceeding recursively, we can prove that $s_{i} \in \mathcal{C}_{i+1}, i=1, \dots, L-1.$
\end{itemize}
\end{remark}

\begin{example}(A lattice Construction $C^\star$) \label{exnonesting} Let $\mathcal{C}=\{(0,0,0,0,0,0),(1,0,1,1,0,1),(0,0,1,0,1,1),$ $(1,0,0,1,1,0),(0,0,0,0,1,0),(0,0,1,0,0,1),(1,0,0,1,0,0),(1,0,1,1,1,1)\} \subseteq \mathbb{F}_{2}^{6},$ with $L=3, n=2.$ One can notice that $\mathcal{C}_1=\{(0,0),(1,0)\} \nsubseteq \mathcal{S}_2(0,0,0,0)=\{(0,0)\}$ and Theorem \ref{coro4} cannot be applied to this case. However, the set $\mathcal{S}=\{(0,0,0,0,0,0),$ $(0,0,1,0,1,1),(0,0,0,0,1,0),$ $(0,0,1,0,0,1)\} \subseteq \mathcal{C}$ and  Theorem \ref{thmcomplete} guarantees that $\Gamma_{C^{\star}}$ is a lattice.
\end{example} 

	Remark \ref{remarkc} leads to the result below regarding the latticeness of the associated Construction C.
	
\begin{corollary}(\textit{Latticeness of associated Construction C}) Let $\mathcal{C} \subseteq \mathbb{F}_{2}^{nL}.$ If $\Gamma_{C^\star}$ is a lattice then the associated Construction C is also a lattice.
\end{corollary}

\begin{proof} If $\Gamma_{C^{\star}}$ is a lattice, then according to Theorem \ref{thmcomplete}, $\mathcal{S} \subseteq \mathcal{C}.$ In the associated Construction C, we make $\mathcal{C}=\mathcal{C}_{1} \times \mathcal{C}_{2} \times \dots \times \mathcal{C}_{L},$ where $\mathcal{C}_{1}, \mathcal{C}_{2}, \dots, \mathcal{C}_{L}$ are the projection codes. Hence, according to the Remark \ref{remarkc}, $\mathcal{S} \subseteq \mathcal{C}$ is equivalent to $\mathcal{C}_{1} \subseteq \mathcal{C}_{2} \subseteq \dots \subseteq \mathcal{C}_{L}$ and the chain being closed under Schur product, which is sufficient to guarantee that $\Gamma_{C}$ is a lattice.
\end{proof}

\section{Minimum Euclidean distance of Construction $C^\star$} \label{secdetermd}

	An important observation is that unlike Construction C, Construction $C^\star$  is not equi-minimum distance, i.e., in general
if the minimum distance $d$ is achieved by a pair of points $x,y \in \Gamma_{C^\star},$ i.e., $||x-y|| = d,$ there may be some other $x' \in \Gamma_{C^\star}$ such that there is no $y' \in \Gamma_{C^\star}$ that makes $||x'-y'||=d.$
	
\begin{example} (A non-equi-minimum distance Construction $C^\star$) \label{exemd} Consider an $L=3$ and $n=1$ Construction $C^\star$ with main code $\mathcal{C}=\{(0,0,0),(1,0,1),(0,1,1),(1,1,0)\}$ $\subseteq \mathbb{F}_{2}^{3}.$ Thus, elements in $\Gamma_{C^{\star}}$ are 
\begin{eqnarray}
\Gamma_{C^\star}=\{0+8z, 5+8z, 6+8z, 3+8z\}, z \in \mathbb{Z}.
\end{eqnarray}

	The minimum Euclidean distance of $\Gamma_{C^\star}$ is $||6-5||=1,$ but if we fix $x'=0 \in \Gamma_{C^{\star}}$ there is no element $y' \in \Gamma_{C^\star}$ such that $||y'||=1.$ 
\end{example}	

	In order to fix notation, we introduce the following two definitions.
	
\begin{definition} (Points in the constellation induced by a fixed $c \in \mathcal{C}$) We denote by $\Gamma_{C^\star}(c),$ the set of points in $\Gamma_{C^\star}$ generated by a fixed element $c \in \mathcal{C} \subseteq \mathbb{F}_{2}^{nL}.$ In other words, $\Gamma_{C^\star}(c)$ is a shift of $2^L\mathbb{Z}^n$ by $c.$
\end{definition}

\begin{definition}(Squared minimum distances) We denote by $d_{\min}^{2}(\Gamma)$ the squared minimum distance between any two distinct points in a constellation $\Gamma \subset \mathbb{R}^n$ and by $d_{\min}^2(\Gamma,0)$ the squared minimum distance of any nonzero element $c \in \Gamma$ to zero. 
\end{definition}

	If $\Gamma_{C^\star}$ is equi-minimum distance, $d_{\min}^{2}(\Gamma_{C^\star})=d_{\min}^{2}(\Gamma_{C^\star},0),$ we know that to each $c \in \mathcal{C} \subseteq \mathbb{F}_{2}^{nL}, c \neq 0$ we associate a unique element $x(c) \in \Gamma_{C^\star} \subset \mathbb{R}^{n}$ in the hypercube $[-2^{L-1},2^{L-1}]^{n},$ which gives the minimum distance of $\Gamma_{C^\star}(c).$ An explicit expression for the nearest constellation point in $\Gamma_{C^\star}(c)$ to zero is 
\begin{eqnarray} \label{eqm}
d_{\min}^{2}(\Gamma_{C^\star}(c),0) = m_1 +2^2m_2 + 3^2m_3 + \dots + {(2^{L -1}-1)}^2 m_{2^{L -1}-1} + {(2^{L-1})}^2 m_{2^{L -1}},
\end{eqnarray}
where each $m_i, i=1, \dots, 2^{L -1}$ are obtained as follows. For $c=(c_{11}, \dots, c_{1n},c_{21}, \dots, c_{2n},  \dots,$ $ c_{L1}, \dots, c_{Ln})$ we consider the $L-$tuples $\tilde{c}_1  = (c_{11}, \dots, c_{L1}),$ $\tilde{c}_2  =  (c_{12}, \dots, c_{L2}), \dots,$ $\tilde{c}_L  = (c_{1n}, \dots, c_{Ln})$ and $m_i, i=1, \dots, 2^{L -1}$ as
\begin{eqnarray} \label{mj}
m_i & = & \text{number of} \ L-tuples \ c_j, \ j=1, \dots, n,  \ \text{such that} \ c_j \ \text{is the binary representation} \nonumber \\
& & \text{of $i$ or the binary representation of} \ 2^{L-1}-i.
\end{eqnarray}

	To be more specific,	
\begin{align*}
m_1 & = \text{number of} \ c_i \ \text{such that} \ c_{i}=(1,0,0, \dots, 0) \ \text{or} \ c_i=(1,1, \dots, 1), & \nonumber \\
m_2 & = \text{number of} \ c_i \ \text{such that} \ c_{i}=(0,1,0, \dots, 0) \ \text{or} \ c_i=(0,1, \dots, 1), & \nonumber \\
m_3 & = \text{number of} \ c_i \ \text{such that} \ c_{i}=(1,1,0, \dots, 0) \ \text{or} \ c_i=(1,0,1, \dots, 1), & \nonumber \\
\vdots & ~~~~~~~~~~~~\vdots
\end{align*}
\begin{align}
m_{2^{L-1}-1} & = \text{number of} \ c_i \ \text{such that} \ c_{i}=(1,0, \dots, 1, 0) \ \text{or} \ c_i=(1,1,\dots, 0, 1, )& \nonumber \\
m_{2^{L-1}} & = \text{number of} \ c_i \ \text{such that} \ c_{i}=(0,0,0, \dots, 0,1).&
\end{align}

	Note that the second choices have the same coordinates as the first ones up to the first nonzero coordinate and after that, all coordinates are different. Moreover, $\displaystyle\sum_{i=1}^{2^L -1} m_i =n.$
	
\begin{remark} From the expression above, we can see that, given a codeword $c \in \mathcal{C}$ of weight $\omega(c)=w,$ $d_{\min}^{2}(\Gamma_{C^\star}(c),0) \geq \tfrac{w}{L},$ since the minimum distance will be achieved when the projection codewords of $c$ have the largest number of coincident coordinates as possible. Hence, if the minimum distance of the code is such that $d_{H}(\mathcal{C}) \geq L2^{2L},$ we can assert that $d_{\min}^{2}(\Gamma_{C^\star},0)=2^{2L}.$
\end{remark}

\begin{example}(Minimum distance of Construction $C^\star$) For $L=2$ and $w \geq 32,$ ($n \geq 16$), we have that $d_{\min}^{2}(\Gamma_{C^\star},0)=2^4.$
\end{example}

	A more concise expression for the minimum distance to zero in $\Gamma_{C^\star}$ can also be derived from \eqref{mj}, by observing that for $c=(c_{1}, c_{2}, \dots, c_{L}) \in \mathcal{C},$ $c \neq 0,$ $c_i=(c_{i1}, c_{i2}, \dots, c_{in}),$ $i=1, \dots, L:$
\begin{eqnarray}
d_{\min}^{2}(\Gamma_{C^\star}(c),0)=||2^{L-1}c_L - 2^{L-2}c_{L-1} - \dots - 2c_{2} - c_1||^2.
\end{eqnarray} 
	It follows that
\begin{equation} \label{dminclose}
d_{\min}^{2}(\Gamma_{C^\star},0)= \min_{c=(c_{1}, c_{2}, \dots, c_{L}) \in \mathcal{C} \atop c \neq 0} \{ ||2^{L-1}c_L - \sum_{i=1}^{L-1} 2^{i-1}c_{i}||^{2}, 2^{2L}\}.
\end{equation}

	If $\Gamma_{C^\star}$ is geometrically uniform, the above expression provides a closed-form formula for the minimum distance of $\Gamma_{C^\star},$ otherwise it is an upper bound for this distance. Therefore \eqref{dminclose} presents a closed-form formula for the minimum distance of a $2-$level Construction $C^\star$ and also when $\Gamma_{C^\star}$ is a lattice, for example.
	
	From Equation \eqref{dminclose}, it could be expected that given a code $\mathcal{C} \subseteq \mathbb{F}_{2}^{nL}$ with minimum weight of projection codes $d_H(\mathcal{C}_1), \dots, d_{H}(\mathcal{C}_L),$ a larger minimum distance will be achieved as $d_H(\mathcal{C}_{i})$ increases with $i,$ with $i=1, \dots, L.$ For example, for $L=2$ and weights of projection codes given by $d_{H}(\mathcal{C}_{1})$ and $d_{H}(\mathcal{C}_{2}),$ respectively, if $d_{H}(\mathcal{C}_{2}) > d_{H}(\mathcal{C}_{1}),$ by considering $||2c_2-c_1||^{2} = \langle 2c_2-c_1,2c_2-c_1 \rangle,$ we can derive from \eqref{dminclose} that 
\begin{equation} \label{dminclose2}
d_{\min}^{2}(\Gamma_{C^\star}) \geq \min \{ 4 d_{H}(\mathcal{C}_{2}) -3 d_{H}(\mathcal{C}_{1}), 16\}.
\end{equation}
	
	Regarding to general upper and lower bounds, since $\Gamma_{C^{\star}}$ is a subset of $\Gamma_C,$ $d_{\min}^{2}(\Gamma_{C^{\star}}) \geq d_{\min}^{2}(\Gamma_{{C}}),$ where $\Gamma_{{C}}$ is the associated Construction C (Definition \ref{associated}). A loose and direct upper bound for $d_{\min}^{2}(\Gamma_{C^{\star}})$ is given by:
\begin{equation} \label{bounds}
d_{\min}^{2}(\Gamma_{C^\star}) \leq d_{\min}^{2}(\Gamma_{C^\star},0) \leq d_{\min}^{2}(\overline{\mathcal{S}}) = \displaystyle\min_{d_{H}(\mathcal{S}_{i}(0,\dots,0)) \neq 0} \{2^{2(i-1)} d_{H}(\mathcal{S}_{i}(0,\dots,0)), 2^{2L}\},
\end{equation}
for $i= 1, \dots, L.$

	Next we compare minimum distances of $\Gamma_{C^\star}$ and the associated $\Gamma_C$ for some previous examples using Equations \eqref{dminclose} and \eqref{bounds}.

\begin{example}(Bounds for the minimum distance of Construction $C^\star$) \begin{enumerate}

\item For the Leech lattice presented in Example \ref{exleech}, we have that $d_{\min}^{2}(\Gamma_{{C}}) = \min \{24,32,32,64\}=24$, $d_{\min}^{2}(\overline{\mathcal{S}}) =  \min\{32,32,64\}=32$ as $\mathcal{S}_{1}(0,\dots,0)$ is a null set and $d_{\min}^{2}(\Gamma_{C^\star},0)=32.$ In this case, $d_{\min}^{2}(\Gamma_{C^\star}) = 32.$ 

\item  In Example \ref{exnonesting}, $d_{\min}^{2}(\Gamma_{{C}}) = \min \{1,4,16\}=1$ and $d_{\min}^{2}(\overline{\mathcal{S}}) =   \min\{16\}=16$ as $\mathcal{S}_{1}(0,\dots,0)$ and $\mathcal{S}_{2}(0,\dots,0)$ are null sets. Also, $d_{\min}^{2}(\Gamma_{{C^\star}},0)=5,$ which coincides with $d_{\min}^{2}(\Gamma_{{C^\star}}),$ because in this case Construction $C^\star$ is a lattice.

\item In Example \ref{exemd}, if we consider the associated Construction C, we have $d_{\min}^{2}(\Gamma_{{C}}) = \min \{1,4,16,64\}=1 $, $d_{\min}^{2}(\overline{\mathcal{S}}) =   \min\{64\}=64$ as $\mathcal{S}_{i}(0,\dots,0)$ are null sets for all $i=1,2,3$ and $d_{\min}^{2}(\Gamma_{{C^\star}},0) = 2$. Here, $d_{\min}^{2}(\Gamma_{C^\star})=1.$ 

\end{enumerate}
\end{example}

	To derive a condition that states when Construction $C^{\star}$ have a better packing density than associated Construction C, we observe that both constellations $\Gamma_{C^\star}$ and its associated $\Gamma_{C}$ contains the lattice $2^{L}\mathbb{Z}^{n},$ i.e., $2^L\mathbb{Z}^{n} \subseteq \Gamma_{C^\star} \subseteq \Gamma_{C}.$ If the number of points of $\Gamma_{C^\star}$ and $\Gamma_{C}$ inside the hypercube $[0,2^{L}]^n$ are respectively $|\mathcal{C}|$ and $|\mathcal{C}_{1}| \dots |\mathcal{C}_{L}|,$ where $\mathcal{C}_{i}, i=1, \dots, L$ are the projection codes, we can assert
\begin{equation}
\Delta(\Gamma_{C^\star})= \dfrac{|\mathcal{C}| \ vol\left(B\left(0,\frac{d_1}{2} \right)\right)}{2^{nL}} \ \ \text{and} \ \ \Delta(\Gamma_{C})= \dfrac{|\mathcal{C}_{1}| \dots |\mathcal{C}_{L}| \ vol\left(B\left(0,\frac{d_2}{2} \right)\right)}{2^{nL}}, 
\end{equation}
where $d_1=d_{\min}(\Gamma_{C^\star})$ and $d_2=d_{\min}(\Gamma_{C}).$ Hence, we can write the following remark:

\begin{remark}
\begin{enumerate}
\item $\Delta(\Gamma_{C^{\star}}) \geq \Delta (\Gamma_{C})$ if and only if ${\left(\dfrac{d_{1}}{d_{2}}\right)}^{n} \geq \dfrac{|\mathcal{C}_{1}| \dots |\mathcal{C}_{L}|}{|\mathcal{C}|},$ \\
\item $\rho_{\text{pack}}(\Gamma_{C^\star}) \geq \rho_{\text{pack}}(\Gamma_{C})$ if and only if $\dfrac{d_{1}}{d_{2}} \geq \left(\dfrac{|\mathcal{C}_{1}| \dots |\mathcal{C}_{L}|}{|\mathcal{C}|} \right)^{1/n},$ ~for~$\rho_{\text{pack}}(\Gamma)=(\Delta(\Gamma))^{1/n}.$
\end{enumerate}
\end{remark}

\begin{example}(Comparing packing densities of $2-$level Constructions $C^\star$ and $C$) Let $\mathcal{C} \subseteq \mathbb{F}_{2}^{2n},$ i.e., we are considering a Construction $C^\star$ with $L=2$ levels (which are geometrically uniform). If the minimum distance of the projection codes are $d_{H}(\mathcal{C}_{1})=1$ and  $d_{H}(\mathcal{C}_{2})=4,$ then, according to \eqref{dminclose2}, $d_{\min}^{2}(\Gamma_{C^\star}) \geq \min\{13, 16\}=13$ and $d_{\min}^{2}(\Gamma_{C}) =1.$ From the previous discussion, $\Delta(\Gamma_{C^{\star}}) \geq \Delta (\Gamma_{C})$ if 
\begin{equation}
(13)^{n/2} \geq \dfrac{|\mathcal{C}_{1}| |\mathcal{C}_{2}|}{|\mathcal{C}|}.
\end{equation}
\end{example}

\begin{example}(Packing densities of Constructions $C^\star$ and $C$)  Consider $\Gamma_{C^\star}$ with $L=2$ and $n=4,$ generated by the main code $\mathcal{C}= $ $\{(0,0,0,0,0,0,0,0),(1,1,1,1,1,1,0,0),(0,0,0,0,1,1,1,1),$ $(1,1,1,1,0,0,1,1)\}.$ Observe that from \eqref{dminclose}, $d_{\min}^2(\Gamma_{C^\star})=d_{\min}^{2}(\Gamma_{C})=4$ and $|\Gamma_{C}|/|\Gamma_{C^{\star}}|=2$ and the associated Construction C presents a better packing density in this case.  

	However, if we consider a code $\overline{\mathcal{C}}$ obtained as permutation of the projection codes of $\mathcal{C},$ i.e., $c=(c_1,c_2) \in \mathcal{C}$ if and only if $\overline{c}=(c_2,c_1) \in \overline{\mathcal{C}}$, we can see from \eqref{dminclose} that $d_{\min}^2(\Gamma_{C^\star})=4,$ $d_{\min}^{2}(\Gamma_{C})=2$ and again $|\Gamma_{C}|/|\Gamma_{C^{\star}}|=2.$ Here, $\left(\frac{2}{\sqrt{2}}\right)^{4} > 2$ and $\Gamma_{C^\star}$ has a better packing density.
\end{example}  

	Table \ref{tvv} summarizes density properties of previous examples according to the discussion presented in this subsection. The notation $\diamond$ below indicates those examples which are nonlattice constellations.

\begin{table}[H] \label{tvv}
\caption{Properties of Construction $C^{\star}$ and its associated Construction C}
\centering
 \begin{tabular}{ | c | c | c | c | c | c | c | c | }  
  \hline\noalign{\smallskip} 
  Example & Dimension & $d_{\min}^{2}(\Gamma_{C^{\star}})$ & $d_{\min}^{2}(\Gamma_{C})$ & $\Delta(\Gamma_{C^{\star}})$ & $\Delta(\Gamma_{C})$ & $\rho_{\text{pack}}(\Gamma_{C^{\star}})$ & $\rho_{\text{pack}}(\Gamma_{C})$  \\
  \noalign{\smallskip}\hline\noalign{\smallskip}
 \ref{ex1}$^\diamond$ & 2 & 1 & 1 & $\pi /16$ & $\pi /8$ & $0.4431$ & $0.6266$ \\
 \ref{ceconjecture} & 2 & 4 & 1 & $\pi /4$ & $\pi /8$ & $0.8862$ & $0.4431$ \\
 \ref{exleech} & 24 &32 & 24 & $0.001929$ & $0.00012$ & $0.7707$ & $0.6236$ \\
 \ref{exnonesting} &2 & 5 & 1 & $0.8781$ & $0.7853$ & $0.9209$ & $0.8861$ \\
  \ref{exemd}$^\diamond$ & 1 & 1 & 1 & $0.5$ & $1$ & $0.5$ & $1$ \\
  \noalign{\smallskip}\hline
\end{tabular}
\end{table}	

\section{Asymptotic packing density of random Construction $C^\star$}

	It is common to assess the ultimate potential of coded modulation schemes by looking on their asymptotic random-coding performance \cite{forneytrottchung, gallager65, shannon67}.  In the case of multilevel Constructions C and $C^\star,$ this amounts to using random (linear or nonlinear) binary component codes, and taking both the number of levels $L$ and the block length $n$ to infinity.

	From a sphere-packing viewpoint, a good reference for comparison is the Minkowski bound \cite{cassels71, conwaysloane, zamir2014}, which states that, in each dimension there exists a lattice whose packing efficiency is at least one half, i.e.,
\begin{equation}
\displaystyle\max_{\Lambda \in \mathbb{R}^n} \rho_{\text{pack}}(\Lambda) \geq \frac{1}{2}.
\end{equation}
It is believed that this bound represents the best asymptotically achievable packing efficiency by a lattice or an infinite nonlattice constellation. For AWGN channel coding, the corresponding goodness measures are the Poltyrev unconstrained capacity and the Poltyrev exponent \cite{poltyrev94}. The latter represents the best achievable error exponent over the high-SNR AWGN channel at rates near capacity, or (with its expurgated version) at rates far below capacity \cite{erez04}.

	In this section we use a simple random coding argument to demonstrate that Construction $C^\star$ can achieve the Minkowski bound $\rho_{\text{pack}}(\Gamma_{C^\star})=\frac{1}{2}$ for each block length $n.$ Combining this with the conjecture that the best asymptotic packing efficiency of Construction C is only $\rho_{\text{pack}}(\Gamma_{C}) \approx 0.4168$ (by the Erez multilevel ``Gilbert-Varshamov bound (GVB)-achieving'' coded modulation \cite{uriclass}), we conclude that Construction $C^\star$ is asymptotically superior to Construction C from a sphere packing viewpoint. 
	
	For the AWGN channel, the same random-coding argument implies that
Construction $C^\star$ with Euclidean decoding achieves the Poltyrev capacity
and error exponent\footnote{Euclidean decoding is fundamentally different than (the more practical yet suboptimal) parallel-bit decoding, assumed in the analysis of BICM, where performance is bounded by the BSC error exponent \cite{fab2008,WaFiHu}.}. We thus conclude that Construction $C^\star$ is asymptotically optimal for packing as well as for modulation over the AWGN channel.

\subsection{A modified Loeliger ensemble of Construction $C^\star$}\label{secloeliger}

	Let $\mathbb{Z}_{q}$ denote the ring of integers modulo $q$, i.e., $\mathbb{Z}_q=\{0,1,\dots,q-1\},$ with alphabet size $q = 2^L,$
and define the $q-$ary code $\mathcal{C}_{q^\star}$ in $\mathbb{Z}_{q}^n$ as the set
\begin{equation}\label{cqstar}
\mathcal{C}_{q^\star} = \{(c_1 + 2c_2 + 4c_3 + \dots + 2^{L-1}c_L) \bmod q: (c_1,\dots,c_L) \in \mathcal{C}\}.
\end{equation}
where $\mathcal{C} \subseteq \mathbb{F}_{2}^{nL}$ is the main code that generates $\Gamma_{C^\star}.$ Construction $C^\star$ is then given by lifting the $q-$ary code $\mathcal{C}_{q^\star}$ into the Euclidean space and replicating it by the cubic lattice $q\mathbb{Z}^n.$

	Our asymptotic analysis of Construction $C^\star$ follows the analysis
of the Loeliger ensemble \cite{erez05, loeliger97}, \cite[Sec. 7.9]{zamir2014}. The Loeliger ensemble is used to prove the existence of a lattice which is asymptotically good for packing, covering, modulation, and quantization. It is based on scaling and randomization of a $q-$ary Construction A lattice. That is, we lift a linear $q-$ary code $\tilde{\mathcal{C}}_{q}$ to the Euclidean space, replicate it by $q\mathbb{Z}^n,$ and multiply by a suitable scalar. In the Loeliger analysis, the alphabet size $q$ is taken to be a prime number and the elements of the generator matrix of $\tilde{\mathcal{C}}_{q}$ are drawn independently and uniformly over $\mathbb{Z}_q.$ Although the $q-$ary code $\mathcal{C}_{q^\star}$ as defined in Equation \eqref{cqstar} is not linear modulo $q$ (unless $\Gamma_{C^\star}$ satisfies the latticeness condition of Theorem \ref{thmcomplete}), and although $q = 2^L$ is not a prime number for $L > 1,$ random generation of the main code $\mathcal{C}$ is sufficient to prove asymptotic goodness.

	Specifically, we can prove the existence of a good constellation in a randomized $q-$ary code-based construction (lattice or nonlattice) if it satisfies the following conditions:
	
\begin{enumerate}
\item \label{cond1} each (nonzero) element of the underlying $q-$ary code is uniformly distributed over $\mathbb{Z}_q^n;$
\item \label{cond2} each pair of elements of the $q-$ary code is statistically independent$;$
\item \label{cond3} the constellation is scaled to a fixed point density, independent of $q$ and $n;$
\item \label{cond4} the period\footnote{The period of a constellation $\Gamma$ is the smallest number $m$ such that $\gamma \bmod m \in \Gamma,$ for all $\gamma \in \Gamma.$} of the scaled constellation grows like $\sqrt{n},$ as $q, n \rightarrow \infty;$
\item \label{cond5} the resolution\footnote{Resolution is the largest number $\delta$ such that the constellation $\Gamma \subseteq \delta \mathbb{Z}^n$.} of the scaled constellation goes to zero, as $q, n \rightarrow \infty.$
\end{enumerate}

	To guarantee these properties in Construction $C^\star$, suppose that each bit $c_{ij}$ of the main code $\mathcal{C} = \{c_{ij}: i=1, \dots, M, j=1,\dots,n\} \subseteq \mathbb{F}_{2}^{nL}$ is drawn independently with a fair coin flip. Clearly, the resulting nonlinear random binary code $\mathcal{C}$ induces a $q-$ary code $\mathcal{C}_{q^\star}$ which satisfies Conditions \ref{cond1} and \ref{cond2}, i.e., each element in $\mathcal{C}_{q^\star}$ is uniformly distributed over $\mathbb{Z}_q^n,$ and each distinct pair of elements is statistically independent. Moreover, the two conditions continue to hold even if $\mathcal{C}$ is a linear binary code, where each bit of its $k \times n$ generator matrix is drawn independently with a fair coin flip. See \cite[Section 6.2]{gallagerbook}.

\begin{remark} To gain some insight into these random properties, let us contrast them with randomized Construction C. Specifically, suppose that $\Gamma_C$ is generated by lifting and replication of the $q-$ary code
\begin{equation}\label{cqc}
\mathcal{C}_q = \{c_1 + 2c_2 + 4c_3 + \dots + 2^{L-1}c_L: c_1 \in \mathcal{C}_1,\dots, c_L \in \mathcal{C}_L\},
\end{equation}
where the component codes $\mathcal{C}_1,\dots, \mathcal{C}_L \subseteq \mathbb{F}_2^n$ are drawn at random. While Condition \ref{cond1} above holds, i.e. each element in $\mathcal{C}_q$ is uniform over $\mathbb{Z}_{q}^{n},$ distinct pairs in $\mathcal{C}_q$ are statistically independent only if they do not coincide in some of the levels.
For example, for $L=2,$ while the elements $c_1 + 2c_2$ and $c_1'+2c_2'$ are statistically independent, the elements $c_1 + 2c_2$ and $c_1+2c_2'$ are not, because they share the same least significant bit (LSB) vector $c_1.$ Thus, randomized Construction C fails to satisfy Condition \ref{cond2}.

	Intuitively, statistical dependence between the elements in the randomized $q-$ary code tends to generate closer points after lifting to the Euclidean space, hence a smaller minimum Euclidean distance. In this respect, Construction $C^\star$ better exploits the benefit of multiple levels compared to Construction C. This intuition is further quantified in Section \ref{sec_erez} below.
\end{remark}

	Returning to Construction $C^\star,$ we shall guarantee that the remaining Conditions \ref{cond3}, \ref{cond4} and \ref{cond5} mentioned above hold following the same derivation as for the Loeliger ensemble in \cite[Sec. 7.9]{zamir2014}. Specifically, let $R=k/n,$ $0 < R < 1$ denote the rate of the main code $\mathcal{C},$ and let $M = 2^{nLR}$ denote the number of points in $\mathcal{C}_{q^\star}.$  Then, the scaled constellation $a^\star \Gamma_{C^\star},$ where $a^\star$ is a scalar given by
\begin{equation}
a^\star = \dfrac{2^{LR}}{q} = \dfrac{1}{q^{1-R}}
\end{equation}
has a unit point density,
\begin{equation}
\dfrac{M}{(a^\star q)^n} = 1,
\end{equation}
independent of $q$ and $n,$ as required by Condition 3. Furthermore, the period of $a^\star \Gamma_{C^\star}$ is $q^R$ (instead of $q$ in the unscaled constellation) and the resolution is $\tfrac{1}{q^{1-R}}$ (instead of 1).   If we now let the alphabet size $q$ grow with $n$ like $O(n^{1/2R}),$ then on the one hand the period will grow to infinity like $\sqrt{n}$ (so the intra-coset distance would not dominate the minimum Euclidean distance), and on the other hand the resolution will shrink to zero as $n$ goes to infinity, thus satisfying Conditions \ref{cond4} and \ref{cond5}. 

	In what follows, we summarize a previous known result showing that Construction C does not asymptotically achieve the Minkowski bound, which asserts the aforementioned superiority of Construction $C^\star.$

\subsection{The Erez ``GVB achieving codes'' Construction C}\label{sec_erez}

	In unpublished class notes \cite{uriclass}, Erez computed the packing efficiency of multilevel coded modulation in the limit of an infinite number of levels $L \rightarrow \infty.$ Clearly, in this limit, coded modulation and Construction C are equivalent. Erez assumed that the component binary codes have balanced Hamming distances \cite{conwaysloane}, i.e., the Hamming distance $d_H(\mathcal{C}_{i})$ of $\mathcal{C}_i$ is $4$ times smaller than $d_H(\mathcal{C}_{i-1})$ for $i=2, \dots, L.$ He also admitted that all component codes satisfy the Gilbert-Varshamov bound with equality \cite{gilbert52, varshamov57}. Since the GVB is believed to characterize the best asymptotic tradeoff between coding rate and minimum Hamming distance of a binary code, his computation amounts to the best packing efficiency of Construction C.

	GVB-achieving codes are those whose size is related to their minimum Hamming distance $d=d_{H}(\mathcal{C})$ via 	
\begin{equation} \label{eqgv}
|\mathcal{C}| \geq \dfrac{2^{n}}{|B(d-1,n)|},
\end{equation}
where $B(r,n)$ denotes an $n-$dimensional zero-centered Hamming ball of radius $r,$ which corresponds to the set of all $n$ length binary vectors with Hamming weight smaller than or equal to $r.$ For a large $n,$ $|B(r,n)| \doteq 2^{nH(r/n)}$ and $H(p)=-p\log_2 p - (1-p)\log_2 (1-p)$ is the binary entropy function for $p \in [0, 1].$ 

	GVB achieving codes can be generated by a ``cookie cutting'' greedy construction \cite[pp. 266]{cover06} \cite{gilbert52, guru2010}, or asymptotically for large $n$ using expurgated random binary codes  \cite{gallagerbook} or random linear codes \cite{barg02, varshamov57}

	We now assume that the component codes of Construction C are balanced, where $\alpha_{1}=d_{H}(\mathcal{C}_1)/n$ and $\alpha_i=d_{H}(\mathcal{C}_i)/n={\alpha_1}/{2^{2(i-1)}},$ for $i=2,\dots,L.$ In addition, if we admit that the codes satisfy the GVB with equality, then we obtain a total number of codewords inside the $q^n$ cube given by
\begin{align}
M = M_1 \cdot M_2 \cdot ~ \dots ~ \cdot M_L = \dfrac{2^{nL}}{2^{n[H(\alpha_1)+\dots+H(\alpha_L)]}},
\end{align}
for large $n.$ Hence, the point density is $(2^{n[H(\alpha_1)+\dots+H(\alpha_1/2^{2(L-1)})]})^{-1}$ points per unit volume for large $n,$ and the minimum Euclidean distance of this constellation is $d_{\min}(\Gamma_C)=\sqrt{d_{H}(\mathcal{C}_1)}=\sqrt{\alpha_1 n}$ for large $L.$ Recalling Equation \eqref{eqpackeff}, writing $\rho_{\text{pack}}(\Gamma_C)=\tfrac{d_{\min}(\Gamma_C)/2}{r_{\text{eff}}(\Gamma_C)}$ and considering the asymptotic volume of a unit ball $V_n \approx \left(\tfrac{2\pi e}{n}\right)^{n/2},$ we obtain, for large $L$ and $n,$ the following formula for the packing efficiency
\begin{eqnarray}\label{packc}
\rho_{\text{pack}}(\Gamma_{C}) = \dfrac{\sqrt{\alpha_1 \pi e}}{\sqrt{2} \cdot 2^{H(\alpha_1)} \cdot ~ \dots ~ \cdot 2^{H(\alpha_1/2^{2(L-1)})}}.
\end{eqnarray}

\begin{figure}[H]
\begin{center}
		\includegraphics[height=5.5cm]{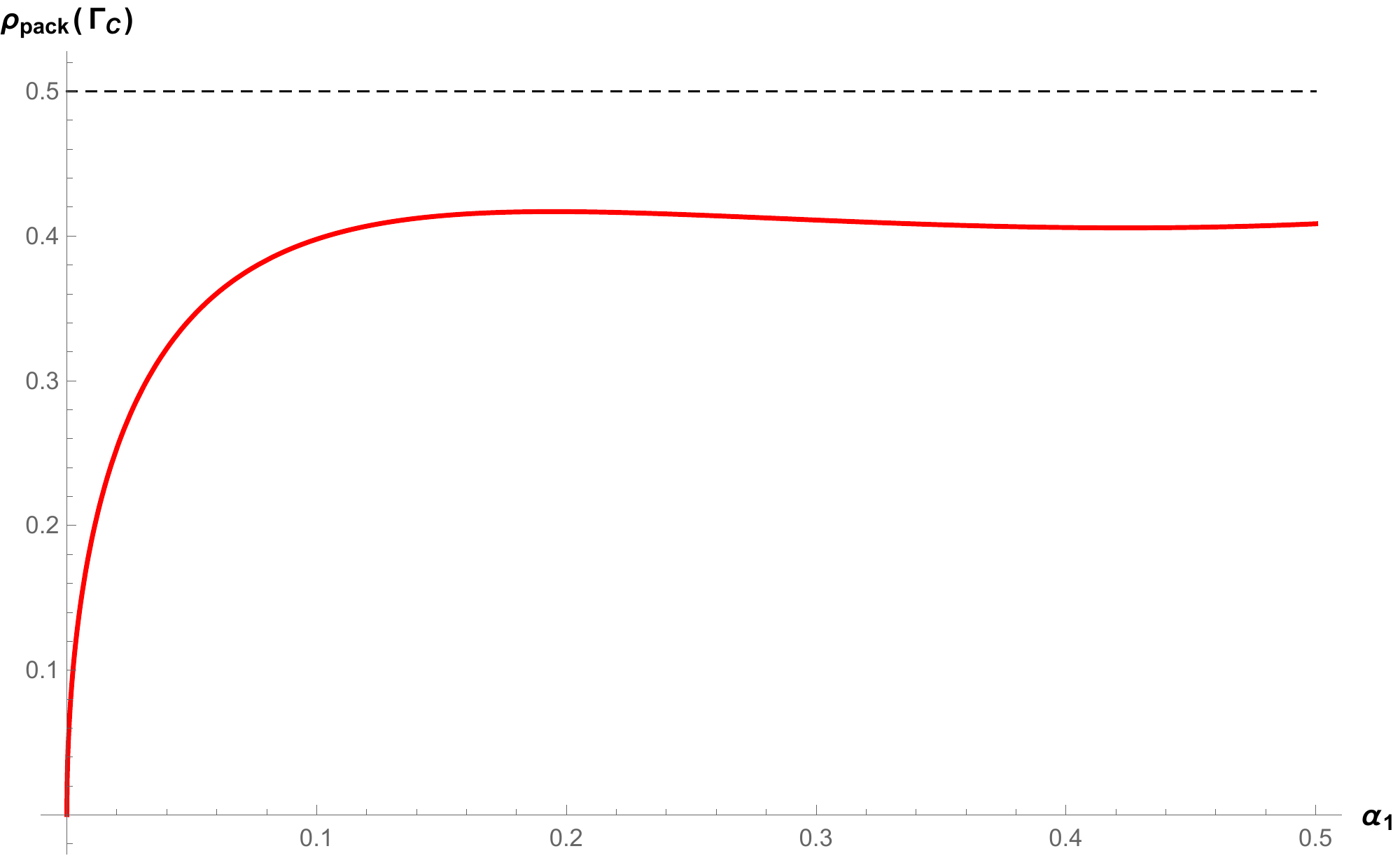}  
\caption{Packing efficiency of Erez ``GVB achieving'' Construction C}
\label{graph}
\end{center}
\end{figure} 

Figure \ref{graph} shows Equation \eqref{packc} as a function of $\alpha_1.$ We can see that the highest value is achieved by $\alpha_1 \approx 0.195$ and a correspondent packing efficiency of $\rho_{\text{pack}}(\Gamma_{C}) \approx 0.4168,$ which is strictly below the Minkowski bound $\rho_{\text{pack}}=\tfrac{1}{2},$ asymptotically achieved by Construction $C^\star,$ as discussed in Section \ref{secloeliger} above.

\section{Conclusion}

	In this paper we provide a detailed investigation about the geometric uniformity of Construction C, including the description of how to produce general geometrically uniform constellations. We introduce a new method of constructing multilevel constellations, denoted by Construction $C^\star,$ which is a generalization of Construction C and is inspired by bit-interleaved coded modulation (BICM). In regard to this construction, we explore some of its properties, including an asymptotic analysis comparing Constructions $C^{\star}$ and C in terms of their packing efficiencies.
	
	Perspectives for future work include changing the natural labeling $\mu$ to the Gray map, which is the standard mapping used in BICM, developing a suitable decoding algorithm for Construction $C^\star,$ taking advantage of the structure of the main code $\mathcal{C} \subseteq \mathbb{F}_{2}^{nL},$ and extending our results to codes defined over a general $q-$ary alphabet. In addition to that, focusing on practical applications, it could be interesting to combine both Constructions $C$ and $C^\star$ in a hybrid scheme. 


%

%


\section*{Acknowledgment}

The authors would like to thank the reviewers and the Associated Editor for their insightful suggestions, and also acknowledge Uri Erez and Or Ordentlich for fruitful discussions about the asymptotic properties of Construction C and $C^\star.$ CNPq (140797/2017-3, 312926/2013-8, 313326/2017-7) and FAPESP (2013/25977-7) supported the work of Maiara F. Bollauf and Sueli I. R. Costa. Ram Zamir was supported by the Israel Science Foundation (676/15).

\ifCLASSOPTIONcaptionsoff
  \newpage
\fi


%
%
%
%
%




\end{document}